\documentclass[conference, twocolumn, 10pt, a4paper]{IEEEtran} 
\usepackage[utf8]{inputenc}

\usepackage{url}

\usepackage{amsmath}

\newenvironment{tightcenter}{%
  \setlength\topsep{0pt}
  \setlength\parskip{0pt}
  \begin{center}}
{\end{center}}

\usepackage{xspace}

\def\EasyCrypt{\textsc{EasyCrypt}\xspace}
\def\prhl{\textsc{pRHL}\xspace}
\def\phl{\textsc{pHL}\xspace}

\usepackage{mathpartir}

\usepackage{textcomp}
\usepackage{listings}

\newcommand{\ensuretext}[1]{\ensuremath{\text{#1}}}

\usepackage{xcolor}

\newcommand{\drarrow}{\raisebox{.07em}{$\mathbin{\scriptstyle\mathsf{-\mkern-3mu>}}$}}
\newcommand{\dlarrow}{\raisebox{.07em}{$\mathbin{\scriptstyle\mathsf{<\mkern-3mu-}}$}}
\newcommand{\drrarrow}{\raisebox{.07em}{$\mathbin{\scriptstyle\mathsf{-\mkern-3mu>\mkern-3mu>}}$}}
\newcommand{\dllarrow}{\raisebox{.07em}{$\mathbin{\scriptstyle\mathsf{<\mkern-3mu<\mkern-3mu-}}$}}

\lstdefinestyle{easycrypt-default}{
  upquote=true,
  escapechar=\#,
  mathescape=false,
  columns=fullflexible,
  keepspaces=true,
  captionpos=b,
  rangebeginprefix={(**\ begin\ },
  rangeendprefix={(**\ end\ },
  rangesuffix={\ *)},
  includerangemarker=false,
  basicstyle=\small\ttfamily,
  identifierstyle={},
  keywordstyle=[1]{\bfseries\color{violet}},
  keywordstyle=[2]{\bfseries\color{olive}},
  keywordstyle=[3]{\bfseries\color{blue}},
  keywordstyle=[4]{\bfseries\color{blue}},
  keywordstyle=[5]{\bfseries\color{red}},
  keywordstyle=[6]{\bfseries\color{violet}},
}

\lstdefinestyle{easycrypt-math}{
  style=easycrypt-default,
  mathescape=true,
}

\lstdefinestyle{easycrypt-math-nocolors}{
  style=easycrypt-default,
  mathescape=true,
  keywordstyle=[1]{\bfseries},
  keywordstyle=[2]{\bfseries},
  keywordstyle=[3]{\bfseries},
  keywordstyle=[4]{\bfseries},
  keywordstyle=[5]{\bfseries},
  keywordstyle=[6]{\bfseries},
}

\lstdefinelanguage{easycrypt-kw}{
  morekeywords=[1]{forall,exists,fun,glob,let,in,var,proc,if,then,else,elif,while,assert,return,res},
  morekeywords=[2]{axiom,hypothesis,axiomatized,lemma,realize,proof,qed,goal,end,import,export,local,declare,hint,nosmt,module,of,const,op,pred,require,theory,section,type,class,instance,print,search,why3,as,Pr,clone,with,prover,timeout,nolocals,Top,equiv,hoare,phoare,islossless},
  morekeywords=[3]{beta,iota,zeta,logic,delta,simplify,congr,change,split,left,right,generalize,case,intros,pose,have,cut,elim,clear,apply,rewrite,rwnormal,subst,progress,trivial,auto,idtac,move,modpath,fieldeq,ringeq,algebra,transitivity,symmetry,seq,wp,sp,sim,skip,call,rcondt,rcondf,swap,cfold,rnd,pr_bounded,bypr,byphoare,byequiv,fel,conseq,exfalso,inline,alias,fission,fusion,unroll,splitwhile,kill,eager},
  morekeywords=[4]{try,first,last,do,strict,expect},
  morekeywords=[5]{exact,assumption,by,reflexivity,done},
  morekeywords=[6]{admit,smt},
}

\def\Why3{\ensuretext{\ls[language=easycrypt]"why3"}}

\def\If{\ensuretext{\ls[language=easycrypt]"if"}}

\def\Return{\ensuretext{\ls[language=easycrypt]"return"}}

\lstdefinelanguage{easycrypt-base}{
  language=easycrypt-kw,
  morekeywords=[1]{arg},
  morecomment=[n][\itshape]{(*}{*)},
  morecomment=[n][\bfseries]{(**}{*)},
  literate=
    {tau}{$\tau$}1
    {sigma}{$\sigma$}1
    {_1}{${}_{\!1}$}1
    {_2}{${}_{\!2}$}1
    {_3}{${}_{\!3}$}1
    {_n}{${}_{\!n}$}1
    {_p}{${}_{\!p}$}1
    {,,,}{$\mathrel{\ldots}$}1
    {...}{$\mathrel{\cdots}$}1
    {->}{$\drarrow$}1
    {<-}{$\dlarrow$}1
    {->>}{$\drrarrow$}2
    {<<-}{$\dllarrow$}2
}

\lstdefinelanguage{easycrypt}{
  language=easycrypt-base,
  style=easycrypt-default,
}

\lstdefinelanguage{easycrypt-math}{
  language=easycrypt-base,
  style=easycrypt-math,
}

\lstdefinelanguage{easycrypt-math-nocolors}{
  language=easycrypt-base,
  style=easycrypt-math-nocolors,
}
 
\lstnewenvironment{easycrypt}[3][]%
  {\lstset{language=easycrypt,caption={#2},label={#3},#1}}%
  {}

\lstnewenvironment{easycryptmath}[3][]%
  {\lstset{language=easycrypt-math,caption={#2},label={#3},#1}}%
  {}

\lstnewenvironment{easycryptfloat}[3][]%
  {\lstset{language=easycrypt,float=tp,caption={#2},label={#3},#1}}%
  {}

\newcommand{\ecinput}[5][]%
{\lstinputlisting[language=easycrypt,linerange={#4},caption={#3},label={#5},#1]{#2}}

\newcommand{\ecinputfloat}[4]%
{\lstinputlisting[language=easycrypt,float=tp,linerange={#3},caption={#2},label={#4}]{#1}}

\def\ls{\lstinline}

\usepackage[linesnumbered,ruled,vlined]{algorithm2e}
\SetAlgoCaptionSeparator{.}
\newcommand{\mycapfnt}[1]{\textsf{\small #1}}\SetAlCapNameFnt{\mycapfnt}

\usepackage{caption}
\usepackage{subcaption}
\usepackage{framed}

\usepackage{booktabs} %
\usepackage{multirow}
\usepackage{makecell}
\usepackage{multicol}
\usepackage{tabularx}
\usepackage{arydshln}

\usepackage{semantic}

\usepackage{amsthm}
\usepackage{amsmath}
\usepackage{amssymb}
\usepackage{mathtools}
\usepackage{cancel}
\usepackage{units}
\usepackage{stmaryrd}
\newtheorem{theorem}{Theorem}
\newtheorem{definition}{Definition}
\newtheorem{notation}{Notation}
\newtheorem{lemma}{Lemma}

\numberwithin{theorem}{subsection}
\numberwithin{definition}{subsection}
\numberwithin{notation}{subsection}
\numberwithin{lemma}{subsection}
\numberwithin{proposition}{subsection}
\numberwithin{corollary}{subsection}

\usepackage{stackengine}

\usepackage{xifthen}

\usepackage{array}
\usepackage{graphicx}
\usepackage{tikz}
\usepackage{tikz-cd}
\usepackage{forest}
\usepackage{soul}
\usetikzlibrary{matrix,calc}
\usepackage{enumitem}

\newcommand{\inlinepub}[4][]{%
\vskip 1em
{\color{white!90!black}\ifthenelse{\equal{#1}{}}{}{\vline width 3pt}}%
\begin{minipage}[c]{0.02\textwidth}\strut\ \end{minipage}%
\begin{minipage}[c]{.90\textwidth}%
\noindent%
{\bf #2} \\ #3 \\ {\em #4}%
\end{minipage}%
}

\let\oldnl\nl%
\newcommand{\nonl}{\renewcommand{\nl}{\let\nl\oldnl}}%

\makeatletter
\newcommand{\oset}[3][0ex]{%
  \mathrel{\mathop{#3}\limits^{
    \vbox to#1{\kern-2\ex@
    \hbox{$\scriptstyle#2$}\vss}}}}
\makeatother

\newcommand{\enc}[1][]{\ifthenelse{\equal{#1}{}}{\mathbf{\mathcal{E}}}{\mathbf{\mathcal{E}}_{#1}}}
\newcommand{\dec}[1][]{\ifthenelse{\equal{#1}{}}{\mathbf{\mathcal{D}}}{\mathbf{\mathcal{D}}_{#1}}}

\newlength{\protocolArrowLength}
\newlength{\partysize}
\newcolumntype{B}[1]{>{\centering\arraybackslash}b{#1}}
\newcolumntype{C}[1]{>{\centering\let\newline\\\arraybackslash\hspace{0pt}}m{#1}}

\newcommand{\eqdef}{\ {\overset{\mathrm{\textnormal\tiny def}}{=}}\ }
\newcommand{\compeq}{\ {\overset{\mathrm{\textnormal\tiny c}}{\equiv}}\ }

\newcommand{\inr}{\in_{\mathrm{R}}}
\newcommand{\getr}{\stackon[0em]{$\;\leftarrow\;$}{$\mathrm{\scriptscriptstyle\$}$}}
\newcommand{\pbr}[1]{\left ( #1 \right )}
\newcommand{\sbr}[1]{\left [ #1 \right ]}
\newcommand{\cbr}[1]{\left \{ #1 \right \}}
\newcommand{\qbr}[1]{\left < #1 \right >}
\newcommand{\set}[1]{\left \{ #1 \right \}}

\newcommand{\abs}[1]{\left| #1 \right|}

\newcommand{\kstar}{\star}
\newcommand{\at}[1]{\left[ #1 \right]} %

\newcommand{\natset}{\mathbb{N}}

\newcommand{\realset}{\mathbb{R}}

\newcommand{\boolT}{{\tt T}}
\newcommand{\boolF}{{\tt F}}
\newcommand{\boolset}{\set{\boolT, \boolF}}
\newcommand{\zoset}{\set{0, 1}}
\newcommand{\head}{\mathsf{H}}
\newcommand{\tail}{\mathsf{T}}
\newcommand{\htset}{\set{\head, \tail}}

\newcommand{\triplet}[3]{\pbr{ #1, #2, #3 }}

\newcommand{\proj}[2]{\pi_{#1}\of{#2}}
\newcommand{\of}[1]{\pbr{#1}}

\newcommand{\alg}[1]{\mathcal{#1}}

\newcommand{\prob}[1]{\mathrm{Pr}\left[ #1 \right]}
\newcommand{\negligible}[1][]{\mu\ifthenelse{\equal{#1}{}}{}{\of{#1}}}
\newcommand{\nonnegligible}[1][]{\varepsilon\ifthenelse{\equal{#1}{}}{}{\of{#1}}}

\newcommand{\vcell}[2]{\multicolumn{1}{#1}{#2}}

\newcommand{\lift}[1][]{\mathcal{L}\ifthenelse{\isempty{#1}}{}{\of{#1}}}

\newcommand{\pwhileskip}{\mathbf{skip}}
\newcommand{\pwhileif}[3]{\mathbf{if}\; #1\; \mathbf{then}\; #2\; \mathbf{else}\; #3}
\newcommand{\pwhileloop}[2]{\mathbf{while}\; #1\; \mathbf{do}\; #2}
\newcommand{\pwhileeval}[1]{{\left\llbracket #1 \right\rrbracket}}
\DeclareMathOperator{\bind}{\mathsf{bind}}
\DeclareMathOperator{\return}{\mathsf{unit}}
\newcommand{\with}[1]{\set{#1}}
\newcommand{\rwrule}[2]{\with{\nicefrac{#1}{#2}}}
\newcommand{\ecproc}[1]{\mathsf{#1}}
\newcommand{\pwhileprog}[1]{\ifthenelse{\isempty{#1}}{\pwhileskip}{\pwhileverticalprogram{#1}}}
\newcommand{\stackstmt}[1]{#1 \\}%
\NewDocumentCommand{\pwhileverticalprogram}{ >{\SplitList{;}} m }{%
  \begin{array}{l}%
    \ProcessList{#1}{\stackstmt}%
  \end{array}%
}

\newcommand{\csgen}[1][]{\ifthenelse{\equal{#1}{}}{\alg{G}}{\alg{G}\of{#1}}}
\newcommand{\cscommit}[1][]{\ifthenelse{\equal{#1}{}}{\alg{C}}{\alg{C}\of{#1}}}
\newcommand{\csverify}[1][]{\ifthenelse{\equal{#1}{}}{\alg{V}}{\alg{V}\of{#1}}}

\newcommand{\sigmagen}[1][]{\ifthenelse{\equal{#1}{}}{\alg{G}}{\alg{G}\of{#1}}}
\newcommand{\sigmacommit}[1][]{\ifthenelse{\equal{#1}{}}{\alg{C}}{\alg{C}\of{#1}}}
\newcommand{\sigmatest}[1][]{\ifthenelse{\equal{#1}{}}{\alg{T}}{\alg{T}\of{#1}}}
\newcommand{\sigmarespond}[1][]{\ifthenelse{\equal{#1}{}}{\alg{R}}{\alg{R}\of{#1}}}
\newcommand{\sigmaverify}[1][]{\ifthenelse{\equal{#1}{}}{\alg{V}}{\alg{V}\of{#1}}}

\newcommand{\experiment}[1]{\mathsf{#1}}
\newcommand{\Exp}[1][]{\ifthenelse{\equal{#1}{}}{\experiment{Exp}}{\experiment{#1\,Exp}}}

\newcommand{\adversary}{\alg{A}}

\newcommand{\advantage}{\mathrm{\mathbf{Adv}}}

\newcommand{\oracle}{\alg{O}}
\newcommand{\dist}{\alg{A}}
\newcommand{\precondition}{{\Psi}}
\newcommand{\postcondition}{{\Phi}}

\newcommand{\leakable}[1]{\widetilde{#1}}
\newcommand{\secret}{\mathrm{H}} %
\newcommand{\leaked}{\mathrm{L}}

\DeclareMathOperator{\isleaked}{\Lambda}

\newcommand{\rnd}[1]{\oset{\text{\fontsize{6}{4}\selectfont\$}}{#1}}
\newcommand{\rndasgn}{\rnd{<-}}
\newcommand{\secasgn}{\hookleftarrow}
\newcommand{\secrnd}{\rnd{\secasgn}}
\newcommand{\ecmem}[1][m]{\mathrm{\mathbf #1}}
\newcommand{\ecmemfamily}{\mathcal{M}}
\newcommand{\inecmem}[1]{\qbr{#1}}
\newcommand{\invariant}{\oset{\secasgn}{I}}
\newcommand{\eager}{\mathsf{Eager}}

\newcommand{\pWhile}{\mathsf{pWhile}}
\newcommand{\pifWhile}{\mathsf{plWhile}}
\newcommand{\tactic}[1]{\mathsf{#1}}
 
\usepackage[switch,mathlines]{lineno}

\makeatletter
\newcommand{\removelatexerror}{\let\@latex@error\@gobble}
\makeatother
 
\title{A Direct Lazy Sampling Proof Technique in Probabilistic Relational Hoare Logic}
\author{\IEEEauthorblockN{Roberto Metere$^\star$}%
\IEEEauthorblockA{University of York, UK}
\and
\IEEEauthorblockN{Changyu Dong$^\star$}%
\IEEEauthorblockA{Guangzhou University, China}
\thanks{$^\star$ Both authors performed part of this work while affiliated with Newcastle University, UK.}
}
\IEEEoverridecommandlockouts %
\begin{document}
\maketitle

\begin{abstract}
Programs using random values can either make all choices in advance (eagerly) or sample as needed (lazily).
In formal proofs, we focus on indistinguishability between two lazy programs, a common requirement in the random oracle model (ROM).
While rearranging sampling instructions often solves this, it gets complex when sampling is spread across procedures.
The traditional approach, introduced by Bellare and Rogaway~\cite{bellare2004game}, converts programs to eager sampling, but requires assuming finite memory, a polynomial bound, and artificial resampling functions.
We introduce a novel approach in probabilistic Relational Hoare Logic (\prhl) that directly proves indistinguishability, eliminating the need for conversions and the mentioned assumptions.
We also implement this approach in the \EasyCrypt theorem prover, showing that it can be a convenient alternative to the traditional method.
\end{abstract}

\section{Introduction}
\label{sec:intro}

A copious amount of cryptographic proofs are based on the concept of indistinguishability between two different programs.
These proofs can be very complex and intricate.
To break down their complexity, intermediate subsequent variations of those programs are created, which differ by very little to one another, e.g., few lines of code.
Once the indistinguishability between all the adjacent constructions is demonstrated, they are finally combined to prove the original statement, as detailed by Shoup~\cite{shoup2004sequences}.
Indistinguishability is a security property that can be proven against a {\em distinguisher} who can interact with an oracle that has been programmed to expose the functionality (but not the memory) of either program, at random.
Then, the distinguisher is challenged to guess which program they interacted with through the oracle.
We provide a formal description in Section~\ref{sec:indistinguishability}.

In this paper, we focus on the particular problem of proving indistinguishability of two (similar) programs that both store values sampled from a distribution, but differ in how the values are internally memorized: in particular, sampling operations are executed in different procedures.
Depending when random choices are made, a probabilistic program is called {\em eager} if all choices are made upfront before its execution, or {\em lazy} if some of them are delayed until they are actually needed.
We consider the case when two programs perform lazy sampling at different levels, i.e., one program is lazier than the other.
An example can be a security property that is based on the construction of a simulator to prove privacy of secrets~\cite{bellare2006security,barthe2009formal,stoughton2017mechanizing,almeida2019machine}, where secrets are not provided to the simulator.
Security is achieved if an efficient distinguisher who can tell the real execution and the simulated execution apart does not exist\footnote{{\em Computational} indistinguishability allows a probabilistic polynomial-time adversary to succeed with a negligible probability.}.
Commonly such simulators would (need to) defer some probabilistic assignments to later or move them to other procedures, de facto acting {\em lazier} than the real program.

\subsection{The problem}
\label{sec:problem}
The difficulty in a formal proof of lazy sampling originates by two aspects.
First, values assigned at runtime do not carry further information, e.g., if the value had originated from a deterministic or a probabilistic assignment.
With an example, if walking into a pub we spot a coin showing head, we know not if the coin was tossed or was simply faced up that way and forgotten on the table.
Second, the programs we want to compare are generated by an adversary who can choose to call the program procedures in any order and any amount; in other words, programs are partially abstract as, even if we know the code inside the procedures, we cannot know how many calls and in what order they are made.
A proof of indistinguishability of two lazy programs, $P_1$ and $P_2$, always shows, at some point, a proof obligation of equivalence, where a value in the memory of one program, $x$, has to be shown equivalent to the result of a probabilistic assignment, i.e., a sampling operation, in the other program.
Proof obligations of indistinguishability are structured as (i) a precondition (a set of hypotheses) on the memory of two programs, (ii) two blocks of code executed by each programs, and (iii) a post-condition on the memory of the programs after the execution of their blocks.
An intuition is illustrated in Fig.~\ref{fig:lazy-sampling-problem}.
\begin{figure}[htbp]
  \includegraphics[width=\columnwidth]{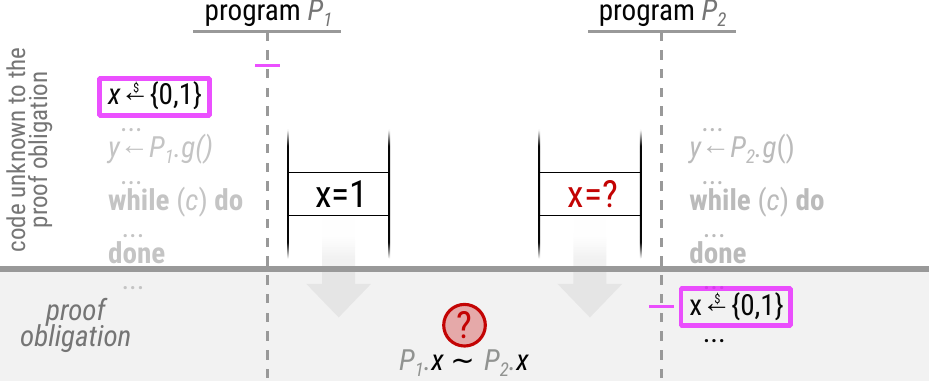}
  \caption{Intuition of how the proof obligation for the lazy sampling problem looks like.}
  \label{fig:lazy-sampling-problem}
\end{figure}
We note that the part of $P_1$'s flow that precedes the current point of the proof (proof obligation) performs the same sampling operation on $x$ (boxed) as $P_2$'s current flow.
Only their effect in memory are known to the proof obligation.
In principle, we could provide some additional hypotheses to the proof obligation about the history of $x$ in the memory of $P_1$; however, the required formality for this is far from trivial at least for two reasons.
First, those hypotheses must refer to memory location and ``preserve'' memory to store additional information.
This translates to the fact that program code can affect those hypotheses in a way that it can easily break soundness.
We develop this concept when we prove the soundness of our strategy, see Section~\ref{sec:soundness}.
Second, the flow of both programs is party abstract, thus one must be able to find hypotheses that are valid for {\em all} possible abstractions.
This point is crucial as, in general and without additional knowledge, the equivalence of a freshly sampled value and a specific value in memory does not hold.

\subsection{Current approach}
\label{sec:state-of-the-art}
Traditionally and currently, lazy sampling proofs are solved in those cases where lazy programs can be transformed into eager programs.
When comparing two eager programs, in fact, the problem described above cannot occur, as the sampling distribution is substituted by randomly-filled, read-only, single-use memory locations before the execution.
This technique was initially proposed by Bellare and Rogaway~\cite{bellare2004game} roughly 20 years ago; they also showed its application to indistinguishability proofs of cryptographic games~\cite{bellare2006security}.
However, they did not provide a formal proof of this technique and, as it often happens in pen-and-paper proofs, many details are omitted or given for granted.
To the best of our knowledge, the first (and only) detailed description of the eager approach has been provided by Barthe et al.~\cite{barthe2009formal}, formalized in the probabilistic Relational Hoare Logic (\prhl).
The \prhl{} is an extension of the Hoare logic where programs are probabilistic and the logic allows for equivalence proofs of programs.
Interestingly, their formalization highlighted the complexity of formally mechanizing lazy samplings with the eager approach: in fact, the technique is based on the injection of an artificial {\em resample} procedure into the original programs, then prove equivalence with interprocedural code motion\footnote{With code motion, statements would be moved within a program, generating a variant that would produce the same functionality as the original program at runtime, despite of the changes.}.
Such formalization actually requires special, complex proof tactics.
The high level description of equivalence of two $P_1$ and $P_2$, denoted as $P_1 \compeq P_2$ is illustrated in Fig.~\ref{fig:contribution-lazy-sampling} and described as follows.
\begin{figure}[htbp]
  \includegraphics[width=\columnwidth]{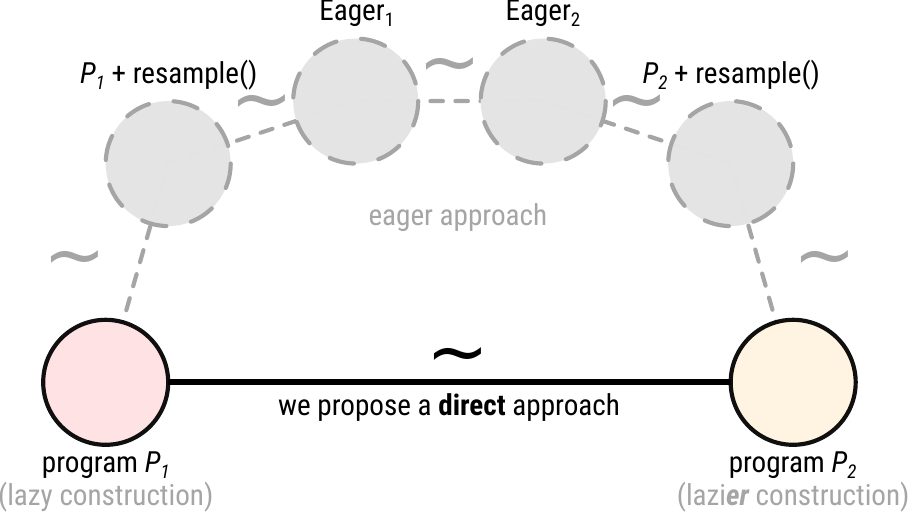}
  \caption{We propose an approach that keeps track of previous sapling operations to allow for direct proofs without artificial intermediate constructions.}
  \label{fig:contribution-lazy-sampling}
\end{figure}
First, $P_1$ and $P_2$ are transformed to $P'_1$ and $P'_2$, each enriched with a {\em resample} procedure, and then converted to $\eager_1$ and $\eager_2$ where all the laziness is removed.
Then one has to prove $P_1 \compeq P'_1$, $P_2 \compeq P'_2$, $P'_1 \compeq \eager_1$, $P'_2 \compeq \eager_2$, $\eager_1 \compeq \eager_2$, and finally compose them.
The equivalence $P_1 \compeq P'_1$ and $P_2 \compeq P'_2$ are trivial, as the artificial resample procedure is not exposed as a functionality to the adversary.
The equivalence $\eager_1 \compeq \eager_2$ can be directly done because values are pre-sampled; so the main problem illustrated in Fig.~\ref{fig:lazy-sampling-problem} is split into both $P'_1 \compeq \eager_1$ and $P'_2 \compeq \eager_2$, where, once you prove either, the other proof should follow the same steps.
Conversions to eager programs require two additional restrictions or assumptions: i) the space of the sampled values must be finite, otherwise an eager construction will possibly never terminate; and ii) especially when complexity is to be considered, the whole eager process has to be bound by a polynomial\footnote{The polynomial bound to the entire eager construction affects its space of the sampled values too, that must be polynomial.}, not to jeopardize the final (computational) result.
The question that naturally arises is: {\em can a direct proof of lazy sampling, that does not require the creation of additional programs, be done?}
Not only we answer positively to this question, but also remove additional assumption required by the eager approach.

\subsection{Our approach and contribution}
\label{sec:contribution}
Our solution differs significantly from the current approach.
Our core idea is to keep track of whether a value associated to a variable has been sampled or not, and from which distribution.
This extra information is recorded along with the value and cannot be directly modified by the program, i.e., it is known only in the environment of execution.
This allows us to specify hypotheses that are valid for {\em all} possible abstractions and, as such, be written in the precondition of the proof obligation that typically arises in lazy samplings, and that we illustrated in Fig.~\ref{fig:lazy-sampling-problem}.

To the best of our knowledge, this is the first time that such a theoretical proof technique for indistinguishability has been proposed and implemented for imperative code.
The problem we address occurs in many indistinguishability proofs: the complexity of eager tactics may discourage cryptographers from attempting to mechanize (part of) their proofs, as they would need considerably more effort to complete them.
With the above considerations, we consider our approach as an elegant alternative solution to the existing one.

Our contribution to the body of knowledge is a twofold:
\begin{itemize}
  \item A novel, alternative, elegant proof strategy to lazy sampling that does not require the theoretical additional assumptions of the eager-lazy approach and allows for direct proofs without intermediate and artificial constructions..
  \item Finally, we implement our strategy in the reasoning core of \EasyCrypt to compare it with an implementation of the eager approach to demonstrate its potential to simplify formal proofs.
\end{itemize}

In the next section, we describe the formality that describes both the problem and our strategy, the \prhl, as well as a formal definition of indistinguishability.
For completeness, we provide a description of the traditional approach in Appendix~\ref{sec:eager-lazy-approach}.
Sections~\ref{sec:ifsupport}~and~\ref{sec:pifwhile} are dedicated to how we extend the language at the basis of the reasoning logic (\prhl) to allow to keep track of additional information along variables, along with dedicated proof tactics to handle the extended language.
We illustrate a proof sketch of the soundness of our approach in Section~\ref{sec:tactics}, and provide a complete, formal proof of our lemmas and theorems in Appendix~\ref{sec:formal-proofs}.
The paper concludes with a re-implementation of the two lazy programs found in the official eager approach in \EasyCrypt and makes a comparison of lines of code and other aspects, to show that our approach can get simpler and shorter proofs.

\section{Preliminaries}
\label{sec:preliminaries}

In this section, we introduce the formality that we use: the cryptographic notions required to formalize indistiguishability, \prhl, as well as the probabilistic Relational Hoare Logic.

\subsection{Indistinguishability}
\label{sec:indistinguishability}
Indistinguishability is a central notion in the theory of cryptography.
We provide a definition based on a cryptographic experiment where a challenger plays against a probabilistic polynomial-time (PPT) adversary, called a distinguisher, who is challenged to tell two probability distributions apart.
For example, the security of an encryption system is defined as the indistinguishability of ciphertexts, where the adversary is asked for two plaintexts and then challenged with the encryption of one of them; the adversary should not be able to guess which plaintext corresponds to what challenge (significantly) better than a coin toss.

The cryptographic experiment of indistinguishability $\experiment{Exp}_{P_0, P_1}^{\dist,\oracle}$, as illustrated in Fig.~\ref{fig:oracle-exp-dist}, relates the two constructions $P_0$ and $P_1$ and can be described as follows.
\begin{enumerate}
  \item The challenger initializes the constructions $P_0$ and $P_1$.
  \item The challenger flips a coin to program the oracle $\oracle$ to use either $P_0$ or $P_1$.
  \item The adversary is allowed to interact with the oracle $\oracle$ a polynomial number of times\footnote{The adversary gets no extra information from the oracle if not the output of the interrogated functions.}, then she sends her best guess to the challenger.
  \item Finally, the challenger outputs whether the adversary, $\dist$, guesses correctly or not.
\end{enumerate}
\begin{figure}[htbp]
  \includegraphics[width=\columnwidth]{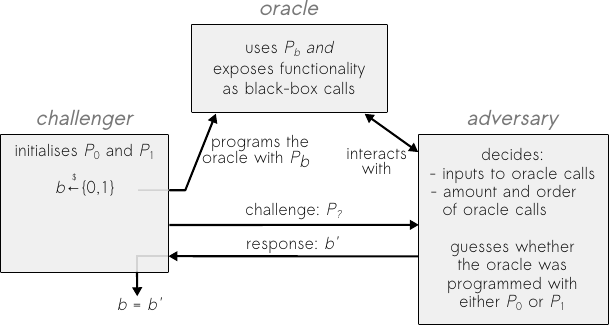}
  \caption{Cryptographic experiment. The adversary $\dist$ is allowed to interact with the oracle $\oracle$ at most a polynomial number of times to tell the two constructions $P_0$ and $P_1$ apart}
  \label{fig:oracle-exp-dist}
\end{figure}
The pseudo-code of a such cryptographic experiment is provided in Fig.~\ref{fig:oracle-exp-dist-pseudocode}.
\begin{figure}[htbp]
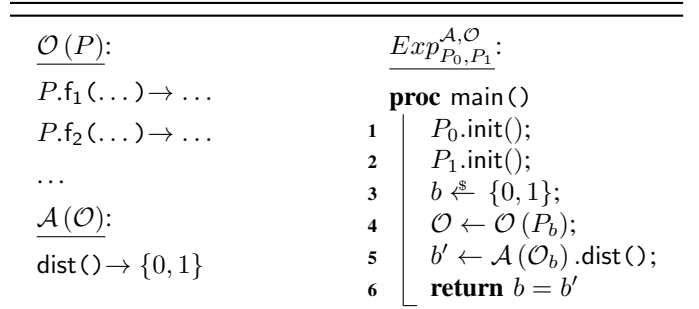

  \removelatexerror
  \begin{algorithm}[H]
    \SetKwFunction{main}{$\mathsf{main}$}{}
    \SetKwFunction{distinguish}{$\mathsf{dist}$}{}
    \SetKwFunction{init}{$\mathsf{init}$}{}
    \SetKwFunction{f}{$\mathsf{f_1}$}{}
    \SetKwFunction{g}{$\mathsf{f_2}$}{}
    \SetKwProg{Procedure}{proc}{}{}
    \SetKwProg{PureProcedure}{}{}{}

    \begin{multicols}{2}
        \SetAlgoNoEnd
        \nonl \underline{$\oracle\of{P}$}: \\
        \BlankLine
        \nonl \PureProcedure{$P$.\f{$\dots$}$ -> \dots$}{}
        \nonl \PureProcedure{$P$.\g{$\dots$}$ -> \dots$}{}
        \nonl \dots
        \BlankLine

        \nonl \underline{$\dist\of{\oracle}$}: \\
        \BlankLine
        \nonl \PureProcedure{\distinguish{}$ -> \zoset$}{}

        \newpage
        \SetAlgoShortEnd

        \nonl \underline{$Exp_{P_0, P_1}^{\dist, \oracle}$}: \\
        \BlankLine
        \nonl \Procedure{\main{}}{
          $P_0.\init()$; \\
          $P_1.\init()$; \\
          $b \getr \zoset$; \\
          $\oracle \gets \oracle\of{P_b}$; \\
          $b' \gets \dist\of{\oracle_b}.\distinguish{}$; \\
          \Return $b = b'$ \\
        }
    \end{multicols}
      \BlankLine %
  \end{algorithm}
  \caption{Pseudo-code of the cryptographic experiment.}
  \label{fig:oracle-exp-dist-pseudocode}
\end{figure}
As security is defined against all PPT adversaries, the interaction between the distinguisher $\dist$ and the oracle $\oracle$ is captured abstractly, i.e., $\oracle$ is passed to $\dist$ as an argument so that $\dist$ can interact with $\oracle$ according to any strategy.
The adversary is aware of the security parameter $n$ of reference, so we write $\dist$ as a short notation for $\dist\of{1^n}$.
We remark that our indistinguishability experiment and our definition of computational indistinguishability are a generalization of common indistinguishability experiments~\cite{shoup2004sequences,katz2014introduction}, and are based on the textbook definition of {\em computational indistinguishability of ensembles} by Katz and Lindell~\cite{katz2014introduction}, with the only change that the ensembles are generated by an indistinguishability experiment involving oracles, cf. Appendix~\ref{sec:probability-ensembles}.

\begin{definition}[Advantage]
\label{def:advantage}
Given the experiment $\Exp_{P_0, P_1}^{\dist,\oracle}$ of indistinguishability of two programs $P_0$ and $P_1$ that use the same functionality run against the adversary $\dist = \dist\of{1^n}$, then the distinguishing advantage $\advantage$ of the adversary is defined as
\[
  \advantage^{\Exp^{\dist,\oracle}}_{P_0,P_1}\of{1^n} \eqdef \left| \prob{\dist\of{\oracle_0} = 1} - \prob{\dist\of{\oracle_1} = 1} \right|.
\]
where $\oracle_0$ and $\oracle_1$ are the oracle programmed with $P_0$ and $P_1$ respectively.
\end{definition}
The events $\sbr{\dist\of{\oracle_b} = b'}$ happen when the adversary interacts with the oracle $\oracle_b$ and outputs $b'$, i.e. guesses correctly in the case $M_{b'}$ programmed the oracle, incorrectly in the case $M_{1 - b'}$.
The concept of {\em advantage} allows us to easily define the computational indistinguishability if related to a negligible function as the upper bound.
\begin{definition}[Computational indistinguishability with oracles]
\label{def:oracle-indistinguishability}
Given the experiment $\experiment{Exp}_{P_0, P_1}^{\dist, \oracle}$ of indistinguishability over two (terminating) constructions $P_0$ and $P_1$ of the same functionality and run against the adversary $\dist$, then the probabilistic ensembles $\alg{X}$, interacting with an oracle $\oracle$ programmed with $P_0$, $\oracle_0$, and $\alg{Y}$, interacting with $\oracle$ when programmed with $P_1$, $\oracle_1$, are {\em computationally indistinguishable}, denoted as $\alg{X} \compeq \alg{Y}$, if and only if for any probabilistic polynomial-time adversary $\dist$ exists a negligible function $\negligible$ such that
\[
  \advantage^{\experiment{Exp}^{\dist, \oracle}}_{P_0,P_1}\of{1^n} \leq \negligible[n],
\]
where $n$ is the security parameter, and $\dist$ has no access to the internal state of the oracle, which is programmed uniformly at random with $P_0$ or $P_1$.
\end{definition}
For simplicity, the constructions themselves are said to be indistinguishable, explicitly omitting reference to both their probabilistic ensembles and the oracles, with the notation
\[
  P_0 \compeq P_1.
\]
In the proof we show in Section~\ref{sec:lazy-sampling-sample-usage}, we prove perfect indistinguishability, i.e. $\advantage^{\experiment{Exp}^{\dist, \oracle}}_{P_0,P_1}\of{1^n} = 0$.
This choice is driven by two reasons.
First, indistinguishability proofs where the advantage is negligible have exactly the same structure apart from the negligible case to be handled separately, but this fact is independent of the approach (eager-lazy or direct).
Second, the official \EasyCrypt construction with showcases the eager-lazy approach proves perfect indistinguishability as well; thus, it is easier to compare with.

We remark that stating only perfect indistinguishability definitions (as we need) would have required more changes from the former definitions than a simple substitution in Definition~\ref{def:oracle-indistinguishability}, i.e. $\negligible\of{n} = 0$.

\subsection{Describing algorithms in \prhl: the $\pWhile$ language}
\label{sec:lazy-sampling-reasoning-pwhile}

The \prhl reasons over programs in imperative code.
Programs lie inside modules, that are containers of global variables and procedures.
Procedures are written in the $\pWhile$ language, see Fig.~\ref{fig:pwhile}, and capture the idea of algorithms running in a memory $\ecmem \in \ecmemfamily$ as their execution environment, where $\ecmemfamily$ denotes the set of memories.
\begin{figure}[htbp]\hrulefill
  \[
    \begin{array}{rl l}
      C ::= & \pwhileskip        & \quad\mbox{\small no operation} \\
          | & V <- E             & \quad\mbox{\small deterministic assignment} \\
          | & V \rndasgn \Delta  & \quad\mbox{\small probabilistic assignment} \\
          | & C;\, C             & \quad\mbox{\small sequence} \\
          | & \pwhileif{E}{C}{C} & \quad\mbox{\small conditional branch} \\
          | & \pwhileloop{E}{C}  & \quad\mbox{\small while loop} \\
    \end{array}
  \]
  \caption{Syntax of the $\pWhile$ language. $V$ are variables bound in the memory of $C$, $E$ is an expression, $\Delta$ denotes a {\em distribution} expression.}
  \label{fig:pwhile}
\end{figure}
The semantics of an expression is standard and evaluates to deterministic values in a given memory, where we assume they are well-formed, e.g., an assignment $v <- e$ is well-formed only if the types of $v$ and $e$ match.
Distributions over a set $X$, that we denote as $\Delta_X$, are described by {\em distribution expressions} and follow a monadic formalization~\cite{audebaud2009proofs} that we do not (need to) fully describe.
For our purposes, it suffices to consider that a distribution expression evaluates to an abstraction, a lambda, on a memory that describes a {\em discrete} distribution over a specified typeset, representing a (sub-)probability measure.

Since we reason in a probabilistic setting, the execution, or evaluation, of a program $c$ in an initial memory yields a (sub-)probability distribution over program memories, that we denote as $\pwhileeval{c}: \ecmemfamily -> \Delta_{\ecmemfamily}$.
We use the same notation for the evaluation of expressions (whose results are values), distribution expressions (whose results are distributions over typesets) and execution of commands (whose results are distributions over memories), $\pwhileeval{\bullet}$, where typing disambiguates the context.
The denotational semantics of commands of the $\pWhile$ language is shown in Fig.~\ref{fig:pwhile-semantics}.
\begin{figure}[htbp]\hrulefill
  \[
    \begin{split}
      \pwhileeval{\mathbf{skip}} \ecmem            & = \return \ecmem \\
      \pwhileeval{v <- e} \ecmem                   & = \return \ecmem\rwrule{\pwhileeval{e} \ecmem}{v} \\
      \pwhileeval{v \rndasgn d} \ecmem             & = \bind \of{\pwhileeval{d} \ecmem} \of{\lambda f. \return \ecmem\rwrule{f}{v}}\\
      \pwhileeval{c_1; c_2} \ecmem                 & = \bind \of{\pwhileeval{c_1} \ecmem} \pwhileeval{c_2} \\
      \pwhileeval{\pwhileif{e}{c_1}{c_0}} \ecmem   & = \pwhileeval{c_{\pwhileeval{e} \ecmem}} \ecmem \\
    \end{split}
  \]
  \caption{Selection of denotational (monadic) semantics of $\pWhile$~\cite{barthe2009formal}; we omit the loop as we do not make special mention of it in this paper.}
  \label{fig:pwhile-semantics}
\end{figure}

\subsection{Reasoning in the \prhl}
\label{sec:lazy-sampling-reasoning-pRHL}

Once programs are written in the $\pWhile$ language, one can reason about their deterministic (HL) or probabilistic (\phl and \prhl) behavior in the environment of a proof.

We borrow our notation for the semantics of the \prhl judgments from~\cite{benton2004simple,barthe2009formal,barthe2012probabilistic}.
Let us consider a probabilistic program $c$ and a pre and a post condition $\precondition$ and $\postcondition$.
Both the precondition and the post-condition can refer to variables in the memory and contain relations about them; additionally, the post-condition can refer to the return value of the procedures.
An HL judgment is a Hoare triplet, i.e., the $\postcondition$ is entailed by $\precondition$ and the execution of $c$, where $c$ is a probabilistic program.
As $c$ is probabilistic, $\postcondition$ may not be entailed always, but in relation to some probability $p$: \phl judgments relate their post-condition to a probability measure.
Both judgments must be valid for all memories where a program $c$ is executed, and they are denoted as
\[
  \begin{split}
    |= c : \precondition => \postcondition                & \qquad\mbox{HL  judgment} \\
    |= c : \precondition => \postcondition \;\lozenge\; p & \qquad\mbox{\phl judgment} \\
  \end{split}
\]
where $p$ is a real expression of probability and $\lozenge$ is a logic operation among $<$, $\leq$, $=$, $\geq$, and $>$.
For example, termination can be formally expressed as {\em losslessness}, that is the probability that the post condition $\boolT$ (true) holds after running $c$ is $1$,
\[
  |= c : \boolT => \boolT = 1.
\]

Let us consider two probabilistic procedures $c_1$ and $c_2$ %
and a pre and a post condition $\precondition$ and $\postcondition$.
The conditions are binary relations that can refer to memories where the procedures run.
We denote by $A_{\at{\ecmem_1,\ecmem_2, \dots}}$ the validity of the relation $A$ whose propositions can relate to the memories $\ecmem_1, \ecmem_2$ and so on; we also denote by $c\inecmem{\ecmem}$ a procedure $c$ running in the memory $\ecmem$; in relational judgments, we implicitly call $\ecmem_1$ and $\ecmem_2$ the memories relating to the leftmost and rightmost programs respectively.

In a \prhl judgment, the precondition is a relation over the memories as environments for the programs before their execution.
The post-condition, instead, does not directly relate to memories modified after the execution, but rather to the distribution of such memories.
In other words, the post-condition must be lifted to a relation over distributions over memories.
We denote the lifting of relations with $\mathcal{L}$.
For our purposes, we use this concept exactly as in Barthe et al.~\cite{barthe2009formal} who provided a complete description of its definition and properties.
We do not need to go into such details and we limit our formal introduction to \prhl to the following definition.
\begin{definition}[\prhl judgment~\cite{barthe2009formal}]
\label{def:prhl-judgement}
  We say that two programs $c_1$ and $c_2$ are equivalent with respect to a precondition $\precondition$ and a post-condition $\postcondition$ if and only if
  \begin{equation}
    \begin{split}
      |= c_1 \sim c_2 : \precondition => \postcondition \eqdef & \forall \ecmem_1,\ecmem_2 \in \ecmemfamily. \precondition_{\at{\ecmem_1,\ecmem_2}}  \\
      & => \lift[\postcondition]_{\at{\pwhileeval{c_1} \ecmem_1, \pwhileeval{c_2} \ecmem_2}}\\
    \end{split}
  \end{equation}
\end{definition}
\prhl judgments introduce a way to reason about probabilistic equivalences into the logic.
For example, we can be interested in proving the equivalence of two programs modeling an unbiased coin toss, where one program also flips the coin afterwards.
Their outputs, $r$, directly come from the same probabilistic assignment sampling from the uniform distribution $\htset$, $r \getr \htset$.
To make the two programs different, we flip the coin in the second program, $\ecproc{flip}\of{x} \eqdef \pwhileif{x = \head}{\tail}{\head}$.
Intuitively, flipping the coin does not jeopardize the equivalence of the programs, whose post-condition can be expressed as the equivalence of the program's outputs, $r\inecmem{\ecmem_1} = r\inecmem{\ecmem_2}$, or in short notation, $=_r$.
\[
    |= r \getr \htset \sim \pwhileprog{r \getr \htset; r <- \ecproc{flip}\of{r}} : \boolT => =_r.
\]
The two samplings will not always draw the same element from $\htset$; hence, if we relate the two memories, they would be sometimes equal, sometimes different and, as such, we are not able to prove (or disprove) the equivalence we seek.
Conversely, if we {\em lift} our conclusion to relate to distribution over memories instead, then the two assignments can be soundly treated by a tactic that requires them to draw from the same distribution, thus helping toward proving the equivalence of the two programs.

From the \prhl judgment in Definition~\ref{def:prhl-judgement}, one can derive probability claims to prove the security of cryptographic constructions.
Basically, we want to relate the \prhl judgment to two events, $\sigma_1$ and $\sigma_2$, that refer to the memories $\ecmem_1$ and $\ecmem_2$ in which $c_1$ and $c_2$ respectively run.
If the \prhl judgment and $\postcondition => \sigma_1 => \sigma_2$ are valid, then the judgment is interpreted as an inequality between probabilities, if the precondition $\precondition$ holds for every initial memory $\ecmem_1$ and $\ecmem_2$, denoted as $\precondition_{\at{\ecmem_1,\ecmem_2}}$.
We focus on equality (for perfect indistinguishability), that can be expressed as
\[
  \inference{|= c_1 \sim c_2 : \precondition => \postcondition & \precondition_{\at{\ecmem_1,\ecmem_2}} & \postcondition => \sigma_1 \Leftrightarrow \sigma_2}%
  {\prob{c_1\inecmem{\ecmem_1} : \sigma_1} = \prob{c_2\inecmem{\ecmem_2} : \sigma_2}}.
\]
In the case when $c_1$ or $c_2$ return a value, the events $\sigma_1$ and $\sigma_2$ may involve the return value, which we generically denote as $r\inecmem{\ecmem_1}$ and $r\inecmem{\ecmem_2}$.
The above equality can be used to reason about the {\em indistinguishability} of cryptographic constructions.
It is important to consider that not the whole content of memories is relevant in the post condition.
Also, from the point of view of an external entity, i.e., an adversary to a cryptographic game, only the {\em exposed} parts of the memory are relevant: in other words, it is important to restrict adversaries from accessing local variables used by other programs, e.g., by providing them with oracle calls.

The reasoning core for \prhl judgments offers many tactics, Fig.~\ref{fig:ec-tactics} illustrates the $\tactic{rnd}$ and $\tactic{assign}$, as they are the most relevant to our paper.
\begin{figure}[!ht]\hrulefill
  \[
    \inference[\text{[rnd]}:]%
      {
      \precondition\at{\ecmem_1,\ecmem_2} => \lift[\Theta]\at{\pwhileeval{\delta_1} \ecmem_1,\pwhileeval{\delta_2} \ecmem_2} \\
      \text{where } \Theta\at{v_1,v_2} = \postcondition\at{\ecmem_1\rwrule{v_1}{v},\ecmem_2\rwrule{v_2}{v}}
      }
      {|= v \getr \delta_1 \sim v \getr \delta_2 : \precondition => \postcondition}
  \]
  \[
    \inference[\text{[assign]}:]%
      {
      \precondition\at{\ecmem_1,\ecmem_2} = \postcondition\at{\ecmem_1\rwrule{\pwhileeval{e} \ecmem_1}{e},\ecmem_2\rwrule{\pwhileeval{e} \ecmem_2}{e}}
      }%
      {|= v <- e \sim v <- e: \precondition => \postcondition}
  \]
  \caption{Selection of rules of \prhl~\cite{barthe2009formal}.
  }
  \label{fig:ec-tactics}
\end{figure}
The tactic $\tactic{rnd}$ is used to prove that two sampling operations in two related programs are equivalent (probabilistic assignment); and the tactic $\tactic{assign}$ is used for the same reason but for deterministic assignments.
We will explain in more details and make use those tactics later in Section~\ref{sec:soundness}.

\section{Security labels}
\label{sec:ifsupport}
Our informal idea is to tag a variable with two labels.
One label associates its value with the distribution from where the value was sampled (in the case of a probabilistic assignment).
The other label associated its value with a {\em confidentiality status}, which is either $\secret$, interpreted as {\em unknown to the adversary}, or $\leaked$, interpreted as potentially {\em leaked} to the adversary.
The way we implement this idea is through the extension of a triplet in the language.
For the purpose of indistinguishability, we restrict our efforts to an implementation of all-or-nothing confidentiality, that is the value can either secure or leaked, leaving partial leakage to future extensions of our work.

We now introduce definitions and notation.
We denote the family of all the distributions over any set by $\Delta$, and the set of all the distributions over a generic set $X$ by $\Delta_X$.
Here we define two new types, {\em confidentiality} and {\em leakable}.
The former type models a set $\mathcal{C}$ with only two values, $\secret$ interpreted as {\em secret} and $\leaked$ interpreted as {\em leaked} to the adversary, $\mathcal{C} \equiv \set{\secret, \leaked}$.
The latter type models a family of sets that relate to a generic set $X$ whose elements are labeled with a distribution over $X$ and a confidentiality value:
\[
  \leakable{X} \equiv X \times \pbr{\Delta_X \cup \set{\bot}} \times \mathcal{C},
\]
where $\bot$ is used if the value is not associated with a sampling distribution.
Unions with the $\bot$ value can be easily implemented through option (or maybe) type.
Due to the nature of the leakable type as a native triplet, projection functions are already defined in the language.
For any $\leakable{x} = \triplet{x}{\delta}{\varepsilon} \in \leakable{X}$, we have the three projections:
  $\proj{1}{\leakable{x}} \eqdef x$, $\proj{2}{\leakable{x}} \eqdef \delta$, $\proj{3}{\leakable{x}} \eqdef \varepsilon$.
We define three additional functions over labeled values:
(i) $\isleaked: \leakable{X} -> \boolset$ testing whether a leakable value has been leaked or not,
(ii) $\inr: \leakable{X} -> \Delta_X -> \boolset$ whose output is $\boolT$ if the leakable value is sampled from the provided distribution and $\boolF$ otherwise, and
(iii) $\simeq: \leakable{X} -> \leakable{X} -> \boolset$
modeling the equality of two labeled values ignoring the confidentiality label.
So, let $\leakable{v} = \triplet{v}{\delta_v}{\varepsilon_v}$ and $\leakable{w} = \triplet{w}{\delta_w}{\varepsilon_w}$ be two labeled values over the set $X$ ($\leakable{v}, \leakable{w} \in \leakable{X}$) and $\delta \in \Delta_X$ be a distribution over the same set, then we define
\[
  \begin{split}
            \isleaked {\leakable{v}} \eqdef \varepsilon_v \neq \secret, \quad& \leakable{v} \inr \delta \eqdef \delta_v = \delta, \\
    \leakable{v} \simeq \leakable{w} \eqdef & v = w \land \delta_v = \delta_w. \\
  \end{split}
\]
An extract of the code implementing the feature above described is in Appendix~\ref{sec:ec-code}.

\section{The plWhile language}
\label{sec:pifwhile}

The imperative code that describes programs in the \prhl follows the syntax of the $\pWhile$ language~\cite{barthe2012probabilistic} as described in Figure~\ref{fig:pwhile}, and its semantics is respected by \prhl proof tactics.
Labeling could be implemented in the memory of the interpreter, transparently to the user of the language~\cite{barthe2019verifying}; however, we opted to describe labels as part of the language itself, extending its syntax with two new statements, {\em secure assignment} and {\em secure probabilistic assignment}, dedicated to exclusively manipulating labeled variables.

The reason for our choice is a threefold.
First, we easily protect labels from direct, unsound manipulation, e.g., sampling from a distribution $\delta$ to a variable and not updating its distribution label to $\delta$ is not allowed.
Second, we do not affect the semantics of other $\pWhile$ statements, therefore the soundness of all the theories across the theorem prover cannot be jeopardized by our extension (syntactic sugar).
And third, the user can choose which variables to label and control, thus avoiding superfluous or irrelevant analysis.
Nonetheless, the code produced using the extended syntax is still as clear as before, the usage of labeled values through their dedicated syntax is transparent.
In fact, the two new statements behave as regular deterministic and probabilistic assignments.
In particular, we extended the $\pWhile$ language with two syntax symbols to the $\pifWhile$ language to include labeled variables.
In details, we added $\secasgn$ for {\em secure deterministic assignment} and $\secrnd$ for {\em secure probabilistic assignment}.
A program in the $\pifWhile$ language is defined as the set of commands illustrated in Fig.~\ref{fig:pifwhile} and Fig.~\ref{fig:pifwhile-semantics} (denotational semantics), where the new two syntax elements are framed, and where $\leakable{V}$ is a labeled variable, and $\leakable{E}$ is a basic expression that can be either a labeled variable identifier or a map whose codomain is labeled.
\begin{figure}[htbp]\hrulefill
  \[
    \begin{array}{rll}
      \leakable{C} ::= & C                                                                        & \mbox{\small any $\pWhile$ command} \\
                | & \leakable{C};\, \leakable{C}                                                  & \mbox{\small sequence} \\
                | & \mathbf{if}\; E\; \mathbf{then}\; \leakable{C}\; \mathbf{else}\; \leakable{C} & \mbox{\small conditional branch} \\
                | & \mathbf{while}\; E\; \mathbf{do}\; \leakable{C}                               & \mbox{\small while loop} \\ \hline
    \vcell{|r}{|} & \strut V \secasgn \leakable{E}                                     & \vcell{l|}{\mbox{\small reading from labeled expression}} \\
    \vcell{|r}{|} & \leakable{V} \secrnd \Delta                                        & \vcell{l|}{\mbox{\small writing into labeled variable}} \\\hline
    \end{array}
  \]
  \caption{Syntax of the $\pifWhile$ language extending the $\pWhile$ language (new syntactic operators are framed).}
  \label{fig:pifwhile}
\end{figure}
\begin{figure}[htbp]\hrulefill
  \[
    \leakable{v} = \triplet{v}{\delta_v}{\varepsilon_v} \quad \leakable{e} = \triplet{e}{\delta_e}{\varepsilon_e}
  \]
  \[
    \begin{split}
      \pwhileeval{w \secasgn \leakable{e}} \ecmem      & = \return \ecmem\rwrule{\triplet{e}{\delta_e}{L}}{\leakable{e}}\rwrule{\pwhileeval{e} \ecmem}{w} \\
      \pwhileeval{\leakable{v} \secrnd \delta} \ecmem  & = \bind \of{\pwhileeval{\delta} \ecmem} \of{\lambda f. \return \ecmem\rwrule{\triplet{f}{\delta}{\secret}}{\triplet{v}{\delta_v}{\varepsilon_v}}} \\
    \end{split}
  \]
  \caption{Denotational semantics of $\pifWhile$ for the new syntax.}
  \label{fig:pifwhile-semantics}
\end{figure}

Roughly, $\secrnd$ is syntactic sugar for the probabilistic assignment $\rndasgn$ where the left hand side is additionally labeled as secret, $\secret$.
Afterwards, if a labeled variable is assigned to a regular variable, its security is {\em consumed} by $\secasgn$ that is semantically equivalent to a deterministic assignment $<-$ of the value along with its re-labeling as leaked, $\leaked$.

\section{Proof tactics}
\label{sec:tactics}
Tactics are transformation rules allowing (part of) a theorem, or goal, to mutate into another goal.
The proof is not completed until a final tactic is able to derive the tautology from the current theorem.
The goal structure can be transformed or split into several sub-goals to be independently proven, then finally combined altogether.
In this sense, the proof can be seen as an execution environment of an algorithm.
As introduced in Section~\ref{sec:lazy-sampling-reasoning-pRHL}, conclusions in the \prhl relate to the distribution of memories to reason about probabilities in cryptography.

We describe the three tactics related to the syntax introduced in Section~\ref{sec:pifwhile}, for the special types introduced in Section~\ref{sec:ifsupport}:
\begin{description}
 \item[declassify] makes a variable be labeled as leaked, $\leaked$, independently from its current labeling;
 \item[secrnd]     makes a variable in a probabilistic assignment be labeled as sampled from the sampling distribution, $\delta$, and secret, $\secret$; and
 \item[secrndasgn] works only when two procedures are in relation and makes a probabilistic assignment mutate to a simple assignment if the right conditions are met.
\end{description}
The tactics $\tactic{declassify}$ and $\tactic{secrnd}$ reflect the denotational semantics in Fig.~\ref{fig:pifwhile-semantics}, and work in the HL and \phl similarly, while in the \prhl can be called side by side.
Their rules are described in Fig.~\ref{fig:tactics-hl_phl}.
\begin{figure}[!ht]\hrulefill
  \[
    \inference[\text{[declassify]}:]%
              {
              r: Y & \leakable{x}: \leakable{Y} & \leakable{x} = \triplet{v}{\delta}{\varepsilon}
              \\
              |= \leakable{x} <- \triplet{v}{\delta}{\leaked}; r <- x : \precondition => \postcondition
              }
              {|= r \secasgn \leakable{x} : \precondition => \postcondition}
  \]
  \[
    \inference[\text{[secrnd]}:]%
              {
              \leakable{x}: \leakable{Y} & v : Y &
              v \notin FV\of{\ecmem} \\
              |= v \getr \delta; \leakable{x} <- \pbr{v, \delta, \secret} : \precondition => \postcondition \\
              }%
              {|= \leakable{x} \rnd{\secasgn} \delta : \precondition => \postcondition}
  \]
  \caption{Proof rules in the HL for $\tactic{declassify}$ and $\tactic{secrnd}$. Their correspondence to \phl and two-sided \prhl are not shown because they can be trivially constructed from them.}
  \label{fig:tactics-hl_phl}
\end{figure}
The tactic $\tactic{declassify}$ mutates the syntax of $\secasgn$ into two deterministic assignments: the first labels the right hand value of $\secasgn$ as leaked, and the second assigns the bare value of the labeled variable to the left hand value of $\secasgn$.
The tactic $\tactic{secrnd}$ mutates the syntax of $\secrnd$ into a probabilistic assignment and a deterministic assignment: the first binds a new\footnote{The variable has never been declared or used in the memory before this point in the proof.} variable $v$ in the memory $\ecmem$ where the algorithm runs, then it assigns to $v$ a value sampled from the distribution $\delta$ at the right hand of $\secrnd$, and the second assigns to the left hand the value $v$ labeled with $\delta$ and $\secret$ (secret).
We remark that at this point in a proof, the freshly sampled value must be secret; whether it will remain secret until the end of the proof
depends on the other parameters of the goal, precondition, post-condition and the successive statements.

The rule $\tactic{secrndasgn}$, illustrated in Fig.~\ref{fig:tactics-prhl}, is longer and involves an invariant that must hold before and after calling every corresponding procedure.
The intuition behind it is very simple: it equals the result of a sampling operation from a program to the one from the other program (last judgment).
The second last judgment that requires that the equality of sampled operations in both programs do not jeopardize the validity of the invariant.
All the other requirements are trivial.

When modeling a simulator or a random oracle, the values are usually stored in a map or a table.
We implemented support for maps that are commonly used in \prhl{}.
Such maps model functions where initially all the domain is (virtually) mapped to $\bot$ and it is interpreted as an empty map.
A map $t$ shaping a partial function from $X$ to $Y$ is defined to always reach an option codomain, $t: X -> Y \cup \set{\bot}$.
We use the notation $t\of{x}$ to denote the value of the domain element $x$ in the map $t$.
Map elements, e.g., $t\of{x}$, are treated in the $\pWhile$ language as variables, as, in common, they uniquely point to memory locations.
The domain of definition of the map $t$, denoted as $t_X$, is defined as $t_X \eqdef \set{x \in X | t\of{x} \neq \bot}$.
So an {\em empty map} is easily defined as the empty domain of definition of the map itself, $t = \emptyset \Leftrightarrow t_X = \emptyset$.

We will also make use of an invariant that relates two labeled variables or maps and a distribution; if relating to map, the invariant is universally quantified for all map entries.
Formally, given two maps $t, t': X -> \leakable{Y} \cup \set{\bot}$ and a distribution $\delta \in \Delta_Y$ over the set $Y$, we define the invariant for secure random assignment from the map $t'$ to the map $t$ as:
\begin{equation}
  \label{eq:invariant-maps}
  \begin{split}
  \invariant&\of{t,t',\delta} \eqdef \forall x \in X, \\
            & \left( x \in t'_X => t'\of{x} \inr \delta                                   \right) \\
    \land\; & \left( x \in t_X => x \in t'_X \land \proj{1}{t\of{x}} = \proj{1}{t'\of{x}} \right) \\
    \land\; & \left( x \in t_X => \isleaked {t\of{x}} => t\of{x} = t'\of{x}              \right) \\
    \land\; & \left( x \notin t_X => x \in t'_X => \lnot\isleaked {t'\of{x}}              \right). \\
  \end{split}
\end{equation}
The invariant ensures that for all elements $x$ in the domain $X$ the following properties hold:
(i) if $x$ is set in $t'$, then it is distributed as $\delta$;
(ii) if $x$ is set in $t$, then it is also set in $t'$ and both hold the same sampled value;
(iii) if $x$ is set in $t$ and it has been leaked to the adversary, then $x$ is also set in $t'$ and its image equal that in $t'$, $t\of{x} = t'\of{x}$; and
(iv) if $x$ is set in $t'$ but not in $t$, then the value (in $t'$) is secret.

The proof case supported by the tactic $\tactic{secrndasgn}$ is in a relational proof, where the procedure at the left shows a sampling from the distribution $\delta$ to a labeled variable $\leakable{v}$, followed by an assignment from it, and the procedure at the right shows the same assignment but without any sampling:
\[
  \leakable{v} \secrnd \delta; x <- v \quad\sim\quad x <- v \\
\]
The only case when these two algorithms behave the same is when $v$ in the rightmost algorithm is {\em secret} and its distribution label is $\delta$.
In fact, if a value has been sampled in the past and remains secret, then no output from any other function may have disclosed its content; therefore, we can mutate the sampling in the left algorithm with an assignment whose right value is the one in $v$ from the rightmost algorithm.
As shown in Fig.~\ref{fig:tactics-prhl}, the tactics split the current goal into two sub-goals, identified as the two \prhl{} judgments below the first two lines.
\begin{figure*}[!ht]\hrulefill
    \[
      \inference[\text{[secrndasgn]}:]%
      {
        t\inecmem{\ecmem_1}, t\inecmem{\ecmem_2} : X -> \leakable{Y} \cup \set{\bot} & x : X & \leakable{v}: \leakable{Y} & \delta : \Delta_Y & \leakable{v} \notin FV\of{\ecmem_1} & \ecmem_1^{\prime} = \ecmem_1 \text{ with } \leakable{v} \\
        x\inecmem{\ecmem_1} = x\inecmem{\ecmem_2} & \precondition' =  \pbr{\precondition \land \leakable{v}\inecmem{\ecmem_1'} = t\of{x}\inecmem{\ecmem_2}} \\
        |= \pwhileskip \inecmem{\ecmem_1^{\prime}} \sim \pwhileskip: \precondition' => \postcondition \land \leakable{v}\inecmem{\ecmem_1^{\prime}} = t\of{x}\inecmem{\ecmem_2} \land \lnot \isleaked \leakable{v}\inecmem{\ecmem_1^{\prime}} \land\ x \notin t_X\inecmem{\ecmem_1} \land\ x \in t_X\inecmem{\ecmem_2} \\
        |= \cbr{t\of{x} <- \leakable{v}; r \secasgn t\of{x}}\inecmem{\ecmem_1'} \sim r \secasgn t\of{x} : \precondition' => \postcondition \land t\of{x}\inecmem{\ecmem_1'} \inr \delta \land \invariant\of{t\inecmem{\ecmem_1'}, t\inecmem{\ecmem_2}, \delta} \\
      }%
      {|= t\of{x} \rnd{\secasgn} \delta; r \secasgn t\of{x} \sim r \secasgn t\of{x}: \precondition => \postcondition}
    \]
    \caption{The proof rule in the \prhl for {\em secrndasgn} on maps requires proving two judgments (sub-goals).}
    \label{fig:tactics-prhl}
\end{figure*}
The first goal requires that the value $\leakable{v}$ in the memory $\ecmem_1$ has not been leaked before the execution of the relevant statements: $\lnot \isleaked \leakable{v}\inecmem{\ecmem_1}$ -- this part includes a proof of the invariant.
The second goal mutates the probabilistic assignment into a deterministic assignment with the same exact sampled value in the other construction; again, the invariant is required to hold after this mutation.

In comparison, the theory from Bellare and Rogaway~\cite{bellare2004game} is obliged to re-sample unused values in order to make sure that the distribution over memories matches.
Even if we do not make a practical case on this detail, it is clear that if the value has been used somewhere else, additional care may be required.
For example, if a table stores randomly sampled indexes of another random table~\cite{bost2017forward}, resampling from the former table could cause inconsistency issues in the latter, that, if not amended, might lead to an adversarial strategy to distinguish otherwise indistinguishable programs.
Our strategy cannot be affected by this potential issue, as we do not resample.
Rather, we {\em borrow} the value sampled from the other program in the equivalence, as, conversely to the eager approach, we can be sure that the value is a result of a previous, equivalent sampling operation.
This guarantees that no further usage {\em need} to be adjusted where the stored value has been used\footnote{We remark that the use of a value does not automatically imply that it has been disclosed to the adversary.}.
Borrowing the sampled value from the other program is not particularly surprising, as their probabilistic equality is derived from two values, their distribution of origin and their non-leaked state.
The equality of values follows a similar semantics as the tactic $\tactic{rnd}$.
The rest is an additional requirement enforced by our invariant and the capability to securely label variables.

\subsection{Proof of soundness}
\label{sec:soundness}

Toward soundness, it is of crucial importance to forbid the coder to abuse the new syntax to deliberately modify the labels in the construction itself.
If deliberately labeled, in fact, the tactics introduced would make it possible to prove the false statement, breaking the overall soundness of the reasoning logic.
In practice this may enable automatic SMT solvers to find their way to prove everything, as implied from the assumption $\boolF$.
To avoid such situations, we our implementation forbids any direct or indirect usage (inside other expressions) of variables and maps with labels if not with their dedicated syntax; moreover, no regular tactics are able to discharge statements with labeled values.
The implementation of the above constraints is straightforward: while parsing user code, the {\em leakable} type is detected and checked against its dedicated syntax and allowed usages, an error is thrown otherwise.

The syntactic structures $\secasgn$ and $\secrnd$ that we introduced are simply a shorthand that combines multiple assignments, and the tactics that manipulate them in the proof environment can be seen more as a semantic refactoring rather than a full language extension.
Though, we also prohibit the user to, accidentally or purposely, arbitrarily modify labeled variables through the regular syntax.
As introduced in Section~\ref{sec:tactics}, the two tactics $\tactic{declassify}$ and $\tactic{secrnd}$, illustrated in Fig.~\ref{fig:tactics-hl_phl}, simply rewrite the newly introduced syntax to assignments that are already supported in \prhl{}.
Their soundness is trivially entailed by the language itself.

Conversely to the first two new tactics, the introduction of the tactic $\tactic{secrndasgn}$ requires more attention, as it introduces ambiguity of their meaning: in particular, it is able to mutate a probabilistic assignment into a deterministic assignment.
Let us analyze the (main) goal we want to prove true.
\[
  |= {t\of{x} \rnd{\secasgn} \delta; r \secasgn t\of{x}} \sim r \secasgn t\of{x}: \precondition => \postcondition
\]
The cases we are interested in are those where $x$ is identical in both memories, $=_x$ in the precondition $\precondition$, and the distribution of the results also match, $=_r$ in the post-condition $\postcondition$.
The difference is that the leftmost program samples a value, while the latter does not; so, our goal is {\em not} generally true.
Its provability depends on the content of memories $\ecmem_1$ and $\ecmem_2$: in particular, a sufficient condition for our goal to be true is if $t\of{x}$ is bound to the distribution $\delta$ in the rightmost memory, $\ecmem_2$, as shown by Lemma~\ref{th:virtual-swap-rnd}.
Other distributions, e.g., $\delta' \neq \delta$, could determine the same distribution over the memory, that could be shown through a bijection between $\delta'$ and $\delta$.
However, we limit ourselves to the sufficient case of using the same distribution.

The soundness of our newly introduced tactic $\tactic{secrndasgn}$ is delegated to our ability to artificially map this situation by means of relations over memories.
A high-level intuition on how we do it is illustrated in Fig.~\ref{fig:proof-intuition}.
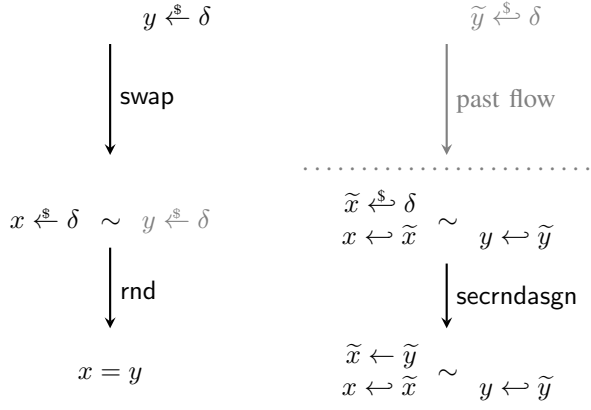
\begin{figure}[!ht]
  \newlength{\nodedistance}
  \setlength{\nodedistance}{2cm}
  \noindent\begin{minipage}[t]{.5\columnwidth}
  \begin{center}
    \begin{tikzpicture}[->,>=stealth,shorten >=1pt,auto,node distance=\nodedistance,thick,node/.style={}]
      \node[node,gray] (flow) {
        $
          \begin{array}{l}
            \color{white}y \secrnd \delta \\
          \end{array}%
            \color{white}\sim%
            \color{black}
          \begin{array}{l}
            y \getr \delta \\
          \end{array}
        $
      };
      \node[node,text width=4cm,align=center,white] (sep) [below of=flow] {
        \dotfill
      };
      \node[node,node distance=.7cm] (secrndasgn) [below of=sep] {
        $
          \begin{array}{r}
            x \getr \delta \\
          \end{array}%
            \sim%
          \begin{array}{l}
            {\color{gray} y \getr \delta} \\
          \end{array}
        $
      };
      \node[node] (rnd_asgn) [below of=secrndasgn] {
        $
          \begin{array}{r}
            \ \\
            \ \\
          \end{array}%
            x = y%
          \begin{array}{l}
            \ \\
            \ \\
          \end{array}
        $
      };

      \path (secrndasgn) edge node [right] {$\tactic{rnd}$} (rnd_asgn);
      \path (flow) [black] edge node [right] {$\tactic{swap}$} (sep);
    \end{tikzpicture}
  \end{center}
  \end{minipage}%
  \begin{minipage}[t]{.5\columnwidth}
  \begin{center}
    \begin{tikzpicture}[->,>=stealth,shorten >=1pt,auto,node distance=\nodedistance,thick,node/.style={}]
      \node[node,gray] (flow) {
        $
          \begin{array}{l}
            \color{white}\leakable{y} \secrnd \delta \\
          \end{array}%
            {\color{white}\sim}%
          \begin{array}{l}
            \leakable{y} \secrnd \delta \\
          \end{array}
        $
      };
      \node[node,text width=4cm,align=center,gray] (sep) [below of=flow] {
        \dotfill
      };
      \node[node,node distance=.7cm] (secrndasgn) [below of=sep] {
        $
          \begin{array}{r}
            \leakable{x} \secrnd \delta \\
            x \secasgn \leakable{x}
          \end{array}%
            \sim%
          \begin{array}{l}
            \ \\
            y \secasgn \leakable{y} \\
          \end{array}
        $
      };
      \node[node] (rnd_asgn) [below of=secrndasgn] {
        $
          \begin{array}{r}
            \leakable{x} <- \leakable{y} \\
            x \secasgn \leakable{x}
          \end{array}%
            \sim%
          \begin{array}{l}
            \ \\
            y \secasgn \leakable{y} \\
          \end{array}
        $
      };

      \path (secrndasgn) edge node [right] {$\tactic{secrndasgn}$} (rnd_asgn);
      \path (flow) [gray] edge node [right] {past flow} (sep);
    \end{tikzpicture}
  \end{center}
  \end{minipage}
  \caption{Intuition of the proof of Lemma~\ref{th:virtual-swap-rnd}.
  Left: sampling is moved ($\tactic{swap}$) to be inline with the corresponding one in the other program.
  Right: the tactic $\tactic{secrndasgn}$ handles labeled variables to overcome the limitation of the tactic $\tactic{rnd}$ that cannot access the ``past flow'', whose statements are not available at this point of the proof.}
  \label{fig:proof-intuition}
\end{figure}
Toward this goal, it is {\em mandatory} that the special labeled values cannot be manipulated (in syntax and in the proofs) in any other way if not through our specialized, controlled tactics.
Full proofs of the lemmas in this section can be found in Appendix~\ref{sec:formal-proofs}.

\begin{lemma}[Virtual $\tactic{swap/rnd}$]
\label{th:virtual-swap-rnd}
  Let $\ecmem_1$ and $\ecmem_2$ be the memories in the leftmost and rightmost programs in a \prhl judgment, $\delta$ a (lossless) distribution expression, and $t$ a map equally defined in both memories.
  Assume that
  \[
    \pwhileeval{t\of{x}} \ecmem_2 = \bind \of{\pwhileeval{\delta} \ecmem_2} \of{\lambda f. \return \ecmem_2\rwrule{f}{t\of{x}}}.
  \]
  then we have
  \[
    |= t\of{x} \getr \delta; r <- t\of{x}  \sim  r <- t\of{x}: =_x => =_r.
  \]
\end{lemma}
The situation depicted in Lemma~\ref{th:virtual-swap-rnd} is not fictional: one can easily reproduce it starting from a provable equivalence of programs where both programs sample, as we show in Appendix~\ref{sec:reproducing-virtual-swap-rnd}.
We now prove the same lemma where the new syntax substitutes the regular syntax.

\begin{lemma}[Virtual $\tactic{secrndasgn}$]
\label{th:virtual-secrndasgn}
  Let $\ecmem_1$ and $\ecmem_2$ be the memories in the leftmost and rightmost programs in a \prhl judgment, $\delta$ a distribution expression, and $t$ a map {\em with labeled codomain} and equally defined in both memories.
  Assume that
  \[
    \pwhileeval{t\of{x}} \ecmem_2 = \bind \of{\pwhileeval{\delta} \ecmem_2} \of{\lambda f. \return \ecmem_2\rwrule{\triplet{f}{\delta}{H}}{t\of{x}}}.
  \]
  then we have
  \[
    |= t\of{x} \rnd{\secasgn} \delta; r \secasgn t\of{x}  \sim  r \secasgn t\of{x}: =_x => =_r.
  \]
\end{lemma}

One might wonder at this point why we did not just use some predicates in the pre- and post-condition to {\em keep track} of memory bindings.
The answer lies in the abstract definition of the adversarial strategy, that is how the security properties are verified {\em for all} adversaries.
In principle, one may attempt and use such predicates for concrete implementations, where calls to external functions (oracles) may manipulate the interested memory regions and can be unfolded into the program to be analyzed.
Indeed, we do use labeling on variables into predicates with exactly the purpose of {\em keep track} of memory bindings.
However this alone is not enough, as the adversary is an {\em abstract} program allowed to query external oracles in an {\em unpredictable} order and times.
This means that one must be able to draw conclusions without having a concrete definition of its program.
As such, there is no direct way to bring along information about when and how probabilistic assignments execute.
Another approach might be that of {\em unfolding} the oracle calls where sampling operations lie by making a case by case proof.
If theoretically one could cover {\em all} possible combinations of calls (assuming that it would suffice to cover a finite amount of them), it is clearly infeasible or not scalable.
We solve this problem by providing a way to {\em keep track} of memory bindings independently of when sampling operations are done, but to do so, we need to allow the related memory to store extra labels.

Lemma~\ref{th:virtual-secrndasgn} shows that the core idea is sound.
Though, it is still not of practical use as its hypothesis directly refers to memory bindings, that are not referable to in \EasyCrypt.
To de-virtualize it, we prove Lemma~\ref{th:practical-secrndasgn}.
\begin{lemma}[De-virtualized $\tactic{secrndasgn}$ hypothesis]
  \label{th:practical-secrndasgn}
  Let $\ecmem_1$ and $\ecmem_2$ be the memories in the leftmost and rightmost programs in a \prhl judgment, $\delta$ a (lossless) distribution expression, $t$ a map with {\em labeled codomain} and equally defined in both memories, and $\invariant$ the invariant as specified in Eq.~\ref{eq:invariant-maps}.
  If we call
  \[
    \begin{split}
    h_1 & \eqdef =_x, h_2 \eqdef \invariant\of{t\inecmem{\ecmem_1}, t\inecmem{\ecmem_2}, \delta}, \\
    h_3 & \eqdef x \notin t_X\inecmem{\ecmem_1}, \text{ and } h_4 \eqdef x \in t_X\inecmem{\ecmem_2}, \\
    \end{split}
  \]
  then
  \begin{equation}
    h_1 \land h_2 \land h_3 \land h_4
  \end{equation}
  implies that
  \[
    \pwhileeval{t\of{x}} \ecmem_2 = \bind \of{\pwhileeval{\delta} \ecmem_2} \of{\lambda f. \return \ecmem_2\rwrule{\triplet{f}{\delta}{H}}{t\of{x}}}.
  \]
\end{lemma}
Now that we showed that we can use the precondition to capture the memory binding related to a sampling operation, we can prove the theorem at the core of our approach, Theorem~\ref{th:if-approach}, that justifies the proof rule in Fig.~\ref{fig:tactics-prhl}.
\begin{theorem}[Direct lazy sampling]
\label{th:if-approach}
  Let $\ecmem_1$ and $\ecmem_2$ be the memories in the leftmost and rightmost programs in a \prhl judgment, $\delta$ a (lossless) distribution expression, $t$ a map with {\em labeled codomain} and equally defined in both memories, and $\invariant$ the invariant as specified in Eq.~\ref{eq:invariant-maps}.
  Assume that
  \begin{equation}
    \begin{split}
    \precondition =\ & =_x \land\ \invariant\of{t\inecmem{\ecmem_1}, t\inecmem{\ecmem_2}, \delta} \\
    & \land x \notin t_X\inecmem{\ecmem_1} \land\ x \in t_X\inecmem{\ecmem_2}, \\
    \end{split}
  \end{equation}
  and let the post condition $\postcondition$ be
  \[
    \postcondition = (=_r)\ \land t\of{x}\inecmem{\ecmem_1} \inr \delta\ \land \invariant\of{t\inecmem{\ecmem_1}, t\inecmem{\ecmem_2}, \delta},
  \]
  then we have
  \[
    |= t\of{x} \rnd{\secasgn} \delta; r \secasgn t\of{x}  \sim  r \secasgn t\of{x}: \precondition => \postcondition.
  \]
\end{theorem}

The result provided by Theorem~\ref{th:if-approach} proves the soundness of the tactic $\tactic{secrndasgn}$ in general, but it does not allow us to make practical use of it.
This is because we did not (need to) specify how we handle the value in $\proj{1}{t\of{x}}$ in the memories, in such a way that it can be used later on in the proofs.
In other words, we proved that the internal memory binds are correct according to the semantics that we defined, but \EasyCrypt needs an explicit way to model it in the proof environment.
In the eager-lazy approach, this aspect was solved by {\em injecting} an artificial re-sample operation in the rightmost program: perhaps, they could hardly do otherwise, as the approach itself was based on sampling upfront.
We notice that, an old re-sampled value could be used (locally) in other data structures, as an index, for example, in other maps (this happens in real-world constructions, e.g., searchable encryption~\cite{bost2018algorithmes}).
Hence, using the eager tactics would create an inconsistent state that can potentially and eventually lead to an adversarial strategy to distinguish otherwise indistinguishable constructions.
With our labels, we can solve the problem without re-sampling.

We know that any variable labeled as secret, $\secret$, has never {\em left} its local environment; in particular, the adversary has no clue about its value.
Thus, as we proved that the content of memories equals, we can simply {\em borrow} the value, that has already been sampled, from the rightmost program to the leftmost.
This idea is inspired by and perfectly aligns with the tactic $\tactic{rnd}$, where equality of sampled values can be stated in a similar fashion.
The technicality to enact this idea is to augment the leftmost memory, $\ecmem_1$, with a {\em fresh} identifier (unknown in the memory $\ecmem_1$) corresponding to a variable whose value is set to $\proj{1}{t\of{x}}\inecmem{\ecmem_2}$.
As we are not even modifying the memories (nor their distribution), we can do so with respect to Theorem~\ref{th:if-approach}.

\section{An example in pRHL proofs}
\label{sec:lazy-sampling-sample-usage}
To support our claim of simplicity of our approach, we consider two (lazy) programs implementing random functions and prove their perfect indistinguishability in the ROM.
We recreate the lazy sampling problem as introduced in Section~\ref{sec:problem} and present in (part of) the proof of the PRP/PRF Switching Lemma completed with the eager technique in~\cite{barthe2010programming}.
Nonetheless, we remark that {\em any} example would have sufficed, as tactics can just be used whenever their premises show up in a (sub-)goal.

The proof structure we implement is common especially in simulation-based proofs of complex constructions, e.g., when a simulator replaces structures depending on (unknown) private inputs with randomly sampled values~\cite{bellare2006security,barthe2009formal,stoughton2017mechanizing,bost2017forward,almeida2019machine}.

\subsection{The two lazy constructions}
Let us consider two constructions, illustrated in Figure~\ref{fig:if-example}, that are described by the two modules $P_1$ and $P_2$\footnote{Their corresponding implementation in \EasyCrypt with labeled maps is shown in Appendix~\ref{sec:lazy-constructions-ec}.}.
Both modules $P_1$ and $P_2$ emulate the behavior of a random function $f: X -> Y$ callable by the procedure $f$.
Additionally they provide a procedure $g$ to preemptively fill the internal map (if not already) with a value that can be later retrieved by calling $f$ on the same argument.
\begin{figure}[!ht]
  \small
  \removelatexerror
  \begin{algorithm}[H]
    \SetKwFunction{init}{$init$}{}
    \SetKwFunction{f}{$f$}{}
    \SetKwFunction{g}{$g$}{}
    \SetKwProg{Procedure}{proc}{}{}

    \begin{multicols}{2}
      \nonl \underline{$P_1$}: {\tt // lazier} \\
      \BlankLine
      \nonl $t: X -> \leakable{Y} \cup \set{\bot}$ \\
      \BlankLine
      \BlankLine
      \nonl \Procedure{\init{}}{
        $t <- \emptyset$; \\
      }
      \nonl \Procedure{\f{$x$}}{
        \If{$x \notin t_X$}{
          $t\of{x} \secrnd \delta$; \\
        }
        $r \secasgn t\of{x}$;\\
        \Return $r$ \\
      }
      \nonl \Procedure{\g{$x$}}{
        \Return \\
      }
      \newpage
      \setcounter{AlgoLine}{0}
      \nonl \underline{$P_2$}: {\tt // lazy} \\
      \BlankLine
      \nonl $t: X -> \leakable{Y} \cup \set{\bot}$ \\
      \BlankLine
      \BlankLine
      \nonl \Procedure{\init{}}{
        $t <- \emptyset$; \\
      }
      \nonl \Procedure{\f{$x$}}{
        \If{$x \notin t_X$}{
          $t\of{x} \secrnd \delta$; \\
        }
        $r \secasgn t\of{x}$;\\
        \Return $r$ \\
      }
      \nonl \Procedure{\g{$x$}}{
        \If{$x \notin t_X$}{
          $t\of{x} \secrnd \delta$; \\
        }
        \Return \\
      }
    \end{multicols}
    \BlankLine
  \end{algorithm}
  \caption{Two constructions $P_1$ and $P_2$; they differ only on the implementation of $g$.}
  \label{fig:if-example}
\end{figure}
Both implementations are lazy: the first time when the input $x \in X$ is used for calling the procedure $f$, a value $y$ is freshly sampled at random from a distribution $\delta$ and stored to an internal map $t$, the next times when $f\of{x}$ will be called, the value will be retrieved its image from $t$.
Additionally to $f$, the adversary can also call the procedure $g: X -> \set{\bot}$ whose signature is identical to $f$ apart from the empty return value, i.e., it produces no output.
The only difference between $P_1$ and $P_2$ is in the implementation of the function $g$, where in $P_1$ does nothing and in $P_2$ actually fills the map.
We notice that if $X$ is finite and $g$ is called upon the whole domain of $t$, the indistinguishability of $f$ is between a lazy and a eager program, reproducing the situation in the proof of the PRP/PRF Switching Lemma in~\cite[Fig.4]{barthe2010programming}.

\subsection{Indistinguishability - Lemma}
The two constructions $P_1$ and $P_2$ are computationally indistinguishable under the Definition~\ref{def:oracle-indistinguishability}.
An oracle $\oracle$ exposes the procedures $f$ and $g$ to be called upon any chosen input argument.
The oracle $\oracle$ is provided to the PPT adversary $\dist$ and can be programmed with either $P_1$ or $P_2$.
The oracle $\oracle$ automatically wraps the construction exposing only $f$ and $g$ and hiding any other procedure, i.e. $init$, and their memory, i.e., $t$ is not accessible.
We can show that, given a probability distribution $\delta$ over a set $Y$, then the advantage of distinguishing between the two constructions $P_1$ and $P_2$ is zero (perfect indistinguishability).
Following the definition of advantage given in Section~\ref{sec:indistinguishability}, the probability of distinguishing either of the constructions has to be the same.
The indistinguishability between $P_1$ and $P_2$ can be therefore captured by the following theorem.
\begin{theorem}
  \label{th:sample-indistinguishability}
  Given the two constructions $P_1$ and $P_2$ as in Figure~\ref{fig:ec-if-example}, a probability distribution $\delta$ over the set $Y$, %
  and the two internal maps of $P_1$ and $P_2$ have been initialized to be empty maps, then
  \[
    \prob{\dist\of{\oracle_1} = 1} = \prob{\dist\of{\oracle_2} = 1}
  \]
  implies that $P_1$ and $P_2$ are indistinguishable, i.e. $P_1 \compeq P_2$.
\end{theorem}
The conclusion of the above theorem relates to the cryptographic experiment $\Exp_{P_1,P_2}^{\dist,\oracle}$ illustrated in Section~\ref{sec:indistinguishability} after a few algebraic steps whose details are discussed in Appendix~\ref{sec:theorem-experiment}.

\subsection{Indistinguishability - Proof}
We now prove the Theorem~\ref{th:sample-indistinguishability} using our novel approach.
The proof can be split into showing that the outputs of the corresponding procedures $f$ and $g$ in $P_1$ and $P_2$, upon the same inputs, are equally distributed.
To keep consistency during the proof, we also use the invariant $\invariant\of{t\inecmem{\ecmem_1},t\inecmem{\ecmem_2},\delta}$ defined in Section~\ref{sec:tactics}.
For simplicity, we denote $t\inecmem{\ecmem_1}$ as $t$ and $t\inecmem{\ecmem_2}$ as $t'$, so the invariant shows the same notation as defined in Section~\ref{sec:tactics}.
Formally, we have to prove the validity of the invariant before calling the distinguisher $\dist$, along with the following pRHL judgments:
\begin{align}
  |=\, & \{ \} \sim \{ \}: \Psi => \Phi \label{eq:subgoal1} \\
  |=\, & P_1.g\of{x} \sim P_2.g\of{x}: \Psi => \Phi \label{eq:subgoal2} \\
  |=\, & P_1.f\of{x} \sim P_2.f\of{x}: \Psi => \Phi \land r\inecmem{\ecmem_1} = r\inecmem{\ecmem_2} \label{eq:subgoal3}
\end{align}
where the precondition $\Psi$ includes the invariant $\invariant\of{t,t',\delta}$ and the requirement of the distiguisher's global variables to equal, and the post-condition includes the same invariant and (the latest judgment only) includes the equality of return values $r\inecmem{\ecmem_1}$ and $r\inecmem{\ecmem_2}$.

The proof of the judgment~\ref{eq:subgoal1}, where no procedures are called, is trivial, as nothing is executed and $\Phi$ is obviously entailed by $\Psi$.
The judgment~\ref{eq:subgoal2}, where $g$ is called, shares similarities with the first, as $P_1.g\of{x} = \{ \}$ is the empty algorithm and both produce no output.
From here, we have two cases
(a) when $x$ is already in the map $t$ of the memory $\ecmem_2$, and
(b), conversely, when $x \notin t'_X$ and about to be sampled.
In the case (a), we reduce to the first judgment~(\ref{eq:subgoal1}), while in (b) we reduce to $|= \{  \} \sim t\of{x} \rnd{<-} \delta : \Psi => \Phi$.

Here, the action of sampling into $t\of{x}\inecmem{\ecmem_2}$ does not jeopardize the validity of the invariant.
In particular, only the last part of the invariant may be affected, i.e., when $x \notin t_X$ but $x \in t'_X$.
So, if we apply the tactic $\tactic{secrnd}$, then we can keep track of the confidentiality of $t'\of{x}$ and prove that $\lnot \isleaked t'\of{x}$, as requested by the invariant.
Finally, the judgment~\ref{eq:subgoal3}, where $f$ is called, is the one that benefits the most from our work.
In all the cases when the value is in or is not in the domains of both $t$ and $t'$, i.e., $x \in t \land x \in t'$ or $x \notin t \land x \notin t'$, the algorithms of either side are exactly the same, and doing consistent operations in both maps does not affect the validity of the invariant.
The most interesting case is when $x \notin t_X$ but $x \in t'_X$.
Differently from the judgment~\ref{eq:subgoal2}, the procedures also produce an output, that corresponds to the values stored in the memories: $t\of{x}\inecmem{\ecmem_1}$ and $t\of{x}\inecmem{\ecmem_2}$.
In particular, we need to show that they are (probabilistically) the same, as in this case the judgment~\ref{eq:subgoal3} reduces to the following:
\[
  \begin{split}
    |= t\of{x} \rnd{<-} \delta; r <- t\of{x} \sim r <- t\of{x} : & \\
      \Psi \land x \notin t_X \land x \in t'_X => \Phi \land  r\inecmem{\ecmem_1} & = r\inecmem{\ecmem_2},
  \end{split}
\]
The labels associated with the value $t\of{x}$ can be used to apply the $\tactic{secrndasgn}$ tactic as explained in Section~\ref{sec:tactics}.
Following the rules for the tactic as illustrated in Figure~\ref{fig:tactics-prhl}, the goal splits into two proof obligations:
\begin{align}
  & |= \{ \} \sim \{ \} : \Psi' => \Psi' \land r\inecmem{\ecmem_1'} = r\inecmem{\ecmem_2} \label{eq:secrndasgn-g1} \\
  & |= t\of{x} <- v; r <- t\of{x} \sim r <- t\of{x} : \Psi' => \Phi' \label{eq:secrndasgn-g2}
\end{align}
where
\begin{align*}
  & \ecmem_1' = \ecmem_1 \mbox{ augmented with } \leakable{v}, \\
  & \Psi' = \Psi \land x \notin t_X \land x \in t'_X \land v\inecmem{\ecmem_1'} = t\of{x}\inecmem{\ecmem_2}, \mbox{ and} \\
  & \Phi' = \Phi \land r\inecmem{\ecmem_1'} = r\inecmem{\ecmem_2} \land t\of{x}\inecmem{\ecmem_1'} \inr \delta.
\end{align*}
The proof obligation~\ref{eq:secrndasgn-g1} holds, as the confidentiality of the value in $\leakable{v}$ is secret as the invariant postulates; this fact holds because the value borrowed from the map was sampled in the past, and it has never been used or revealed.
Finally, the validity of the proof obligation~\ref{eq:secrndasgn-g2} is trivial, as the return values now are simple assignments of the same values on both sides, and the value clearly must have been labeled as sampled from $\delta$.

And the proof is complete.

\section{Comparison with the eager approach}
\label{sec:comparison}

We have shown that our direct approach is a sound theoretical alternative to the eager-lazy approach.
In this section, we compare the proof in Section~\ref{sec:lazy-sampling-sample-usage} as we implemented it in \EasyCrypt with the corresponding proof done with the eager approach.
A practical implementation, in fact, allows us to compare other characteristics that relate to the length and complexity of the proofs in either approach.

To provide the fairest comparison as possible, it is reasonable to pick up an essential implementation of lazy sampling, to avoid any other irrelevant aspects that would spuriously affect the number of proofs and their length.
Not coincidentally, the constructions in Figure~\ref{fig:if-example} are structurally the same as those showcased for the eager approach as recently implemented in the \EasyCrypt official repository, in the file {\tt theories/crypto/PROM.ec}\footnote{See {\tt README2.md} for details in our code, freely available at \url{https://mega.nz/file/gp4XlLLI#q_d6mXUA8DH0ED1_qNnzTJuP8PT_hgrlJTjcymccSPI}.}.%
This case study can be seen as a role model of how to handle two lazy constructions where the problematic proof obligation appears, and it is the one used, for example, in the PRP/PRF Switching Lemma~\cite{barthe2010programming} (that was originally completed in CertiCrypt, predecessor of \EasyCrypt).

What we focus on here is the added simplicity of our theory, related to the problem of lazy sampling.
We could not reliably measure human effort, as the time required for learning either technique, producing a proof (for the first time), or other similar characteristics, as done by Zinzindohou{\'e} et al.~\cite{zinzindohoue2017hacl} for example.
The best we can do is to use crude metrics, i.e., the number and the length of proofs; to do that, we use lines of code, number of words, and non-space characters.

The programs in the eager approach are enriched with two functions, to (deterministically) set and remove elements to and from the map.
Even though they serve as a comprehensive ``history'' of a variable before its re-sampling, we decided to remove them to provide a fairer comparison.
We also stripped out all comments and other lines that might have been incorrectly counted.
The above resulted to a reduced number of source lines of code (SLoC), from $765$ to $432$.
Finaly, we do not take into account how many changes either approach required to be implemented into the \EasyCrypt core\footnote{We changed $37$ files, made $1493$ additions and $202$ deletions.}.
A summarized comparison is illustrated in Table~\ref{tbl:comparison}.
\begin{table}[htbp]
  \caption{Eager-lazy approach compared with our approach based on lazy-sampling labeling.}
  \label{tbl:comparison}
  \centering
  \begin{tabular}{rcccc}
                     &          &           & \bf Non-space & \bf Number of \\
    \em (whole code) & \bf SLoC & \bf Words & \bf Chars     & \bf Proofs \\
    \toprule
    \bf eager approach   & $432$    & $3021$   & $16535$   & $30$ \\
    \bf our approach     & $219$    & $1537$   & $8142$    &  $1$ \\
    reduced effort       & $(49\%)$ & $(49\%)$ & $(51\%)$  &  \\

    \bottomrule
    \\
                      &          &           & \bf Non-space & \bf Length of \\
    \em (proofs only) & \bf SLoC & \bf Words & \bf Chars     & \bf Proofs \\
    \toprule
    \bf eager approach   & $265$    & $1918$   & $10214$   & $75,53,\dots$ \\
    \bf our approach     & $165$    & $1333$   & $7204$    & $165$ \\
    reduced effort       & $(38\%)$ & $(31\%)$ & $(29\%)$  & \\
    \bottomrule
  \end{tabular}
\end{table}
A quick comparison shows a drastic reduction in the number of proofs (from 30 to just 1): this is expected as it is our main contribution.
It is also noticeable a significant reduction in the general effort to complete the proofs in terms of lines, words or characters (about $30\%$).
On the same trend, they {\em need} to create additional constructions (that are not proofs) but account for about another rough $20\%$ of the total effort.
Thus, our approach reaches the indistinguishability in about $50\%$ less effort.
However, the length of our single proof in our approach is more than double the length of the longest proof in the eager-lazy approach.
This is only apparently alarming and does not impact the complexity of our proof.
One reason is that they split the main proof into multiple shorter lemmas and combine them together at the end.
Furthermore, they use a richer, more compact syntax in proofs than us to help them complete the proofs in fewer steps.
We consider the complexity of our proof significantly lower for the following reasons:
(i) our proof is direct, we do not need to create intermediate constructions, prove their indistinguishability, then carry out the main proof;
(ii) the main proof of indistinguishability shows exactly the same structure in both approaches, see Section~\ref{sec:lazy-sampling-sample-usage}: while they decided to split their sub-goals into multiple lemmas, we limit to (visually) split them with comments inside a single proof; nonetheless, the longest sub-goal in our approach accounts for less lines ($64$) than the longest of the proofs in the eager-lazy approach ($75$); finally,
(iii) we do not need to handle interprocedural code motion, that required the implementation of more than a dozen complex (eager-)tactics~\cite{barthe2009formal}, and we do not need re-sampling that, if the re-sampled values have been used, need to be handled to avoid potential inconsistencies in cascade.

\section{Related work}
\label{sec:formal-methods-lazy-sampling}

The problem solved by our proof technique is the same as described in Bellare and Rogaway~\cite{bellare2004game} roughly 20 years ago.
Such technique found its first implementation in the formalization by Barthe et al.~\cite{barthe2009formal} in the \prhl{}, accompanied by an implementation in the tool CertiCrypt.
Nonetheless, some other works share similar objectives as ours: Barthe et al.~\cite{barthe2014probabilistic} prove equivalence with non-interference properties in pre and post conditions of proof obligations (rather than labelling, as we do), while Grimm et al.~\cite{grimm2018monadic} do it in the case when an adversary is static and security cannot be described with oracles.

In particular, Barthe et al.~\cite{barthe2014probabilistic} build on top of the ideas of Swamy et al.~\cite{swamy2011secure} and Nanevski et al.~\cite{nanevski2011verification} and extend refinement types to relational formulae.
In contrast with our approach, they express the non-interference properties in the post condition, rather than directly label the data with security flags, as commonly done in other theories, e.g., information flow analysis.
This choice cannot work in general with unpredictable orders of function calls typical of random oracle proofs, as we require.
They offer a classic implementation of a random oracle that lazily samples (and memoizes) the random function, using a mutable reference holding a table mapping hash queries.
To verify the equivalence properties, they introduce several invariants that must hold and that relate the tables stored by either of the two executions.
As hinted before, they could not automate general direct proofs this way, and their approach requires the {\em manual} insertion (and proof) of program-dependent lemmas.
Their approach, in CertiCrypt, was probably the basis for the eager-lazy implementation that is now available in \EasyCrypt{}, the successor of CertiCrypt.
Conversely, our strategy allows the cryptographer to mechanise and automate a direct proof.

On a different note, Grimm et al.~\cite{grimm2018monadic} provide logical relations for a state monad and use it to prove contextual equivalences.
They show perfect security of one-time-pad, that is very basic and certainly can be described with the \prhl{}.
The objective in common with ours is their ability to prove equivalence of (bijective) random sampling operations.
Interestingly, their proofs are {\em automatically} produced, while we have to do them manually.
However, their work cannot apply nor be extended to equivalences with function calls, as oracles require, that is our case study.

\section{Conclusion and future development}
\label{sec:conclusion}
In this paper, we equip the cryptographer with a novel proof technique for lazy sampling in the random oracle model, that is used to prove the indistinguishability of programs with random samplings, generally referred to as lazy samplings.
We formalized our technique in the \prhl and propose it as an alternative to the already existing eager-lazy approach, originally introduced by Bellare and Rogaway~\cite{bellare2004game}.
We proved the soundness of our theory and implemented it in a theorem prover, \EasyCrypt, that allows for reasoning in the \prhl{}.
Finally, we compared the implementation in the two approaches and show that the proof done with our approach is noticeably simpler and shorter.
However, we remark that different case studies are characterized by different challenges; so, it is ultimately left to the cryptographer the choice of which proof technique would be more convenient for their case and their expertise.

However, this work leaves a couple of unanswered questions.
We compared our proof with the official example showing the eager approach in \EasyCrypt{}, but we do not (i) apply our approach to mechanized proofs other than that, nor (ii) investigate the implementation of our theory in other theorem provers.
To the best of our knowledge, only the Foundational Cryptography Framework (FCF)~\cite{petcher2015foundational}, \EasyCrypt{}, and only recently SSProve~\cite{abate2021ssprove}, implement the \prhl for imperative code; other theorem provers, e.g., CryptHOL~\cite{basin2020crypthol}, deal with probabilistic functions instead.
Even if we have chosen \EasyCrypt, we do not see impediments for FCF or SSProve to implement our strategy as well, and it might be an interesting future work.
Conversely, we expect that different proofs would be required in CryptHOL, as non-deterministic calls to functions are not necessarily captured or dealt with in the same way as in the \prhl{} for imperative programs.

\section*{Acknowledgements} 
\noindent
The authors would like to express their gratitude to Benjamin Grégoire and Gilles Barthe for valuable discussions and generous assistance that significantly contributed to the development of this work.
 
\bibliographystyle{plain}
\bibliography{references}

\appendix

\subsection{Additional preliminaries: negligible functions.}
\label{sec:negligible-functions}

Computational security relaxes standard security properties definitions by letting them fail at most for a negligible probability.
Negligible functions are the necessary formality to capture such a concept.
The family of {\em negligible functions} describes those functions that decrease faster than the inverse to a polynomial.

\begin{definition}[Negligible function]
\label{def:negligible-function}
A function $\negligible : \natset -> \realset$ is {\em negligible} if, for all positive natural $c$, there exists a natural number $n_0$ such that
\[
  \forall n \in \natset,\; n_0 < n => \abs{\negligible[n]} < \dfrac{1}{n^c}.
\]
\end{definition}

Informally, if $n$ is a security parameter, and a security property holds up to negligible probability on $n$, then we can {\em tune} the strength of a security property by increasing or decreasing $n$.
Increasing the security parameter $n$ increases both the time required by honest parties to run the protocol and the breaking algorithm of the attackers.
The key concept is that while the former increases polynomially (as it is an efficient algorithm), the latter increases super-polynomially, e.g. exponentially.

\subsection{Additional preliminaries: probability ensembles and indistinguishability.}
\label{sec:probability-ensembles}
An indistinguishability experiment must capture the security of the protocol for the security parameter $n$ tending to infinite.
The necessary formalism to have this asymptotic approach is given by an infinite sequence of random bitstrings, called {\em probabilistic ensembles}, that can be efficiently sampled.
The following formal notation is adapted from of Hazay and Lindell~\cite{hazay2010efficient}, Goldreich~\cite{goldreich2007foundations}, and Katz and Lindell~\cite{katz2014introduction}. %
\begin{notation}[Probability ensemble]
\label{def:distribution-ensemble}
Given an index value $a \in \zoset^\kstar$, and a security parameter $n \in \natset$, we denote a {\em probability ensemble} $\alg{X}$ indexed by $a$ and $n$ the set of all random variables $X_{a,n}$ for all $a$ and all $n$:
\[
  \alg{X} = \set{X_{a,n}}_{a \in \zoset^\kstar, n \in \natset}.
\]
Additionally, there exists an efficient sampling algorithm $\alg{S}$ such that for all $a$ and $n$, $\alg{S}\of{1^n, a}$ and $X_{a,n}$ are identically distributed.
\end{notation}
The input of the parties in a protocol or in a cryptographic experiment is described by the value $a$ indexing a probability ensemble $\alg{X}$~\cite{hazay2010efficient}, while the description of the output is included in the random variable $X_{a,n} \in \alg{X}$.
The random variable $X_{a,n} \in \alg{X}$ is what the adversary will finally inspect.

We are interested in the indistinguishability between two different constructions $M_0$ and $M_1$ of the same protocol, i.e. they provide the same functionality.
Their output can be seen as two random variables $X_{a,n}$ and $Y_{a, n}$ in the ensembles $\alg{X}$ and $\alg{Y}$ respectively to the games.
The result of the adversarial strategy, or any other event capturing security, can be seen as the event $E$ as a function over one such variable.
So, we can relate $M_0$ and $M_1$ through the probability of events $E_0$ and $E_1$, whose (absolute) difference is at the basis of the concepts of indistinguishability:
\[
  \left| \prob{E_X\of{X_{a,n}}} - \prob{E_Y\of{Y_{a,n}}} \right|,
\]
where $X_{a,n}$ and $Y_{a,n}$ are random variable in the ensembles $\alg{X}$ and $\alg{Y}$.
The event we generally refer to in cryptographic lemmas relate to the output of an adversary to be able to determine a predicate over the {\em challenging} string.
Finally the {\em event} is the experiment outputting $1$ whenever the adversary guesses correctly.
So we write $\prob{\adversary\of{X_{a,n}} = y}$ for the probability of the event when the adversary is provided with $X_{a,n}$ and outputs $y$.
We equivalently write $\prob{\adversary\of{M} = y}$ for $\prob{\adversary\of{X_{a,n}} = y}$ if $X_{a,n}$ relates to the construction $M$.

The ability of the adversary, or distinguisher, to guess is measured by her probability to win a distinguishing game, which is called the {\em advantage} of the adversary.
Depending on the security property and the adversarial capabilities, the adversary can be challenged to distinguish two constructions either (i) by inspecting the transcript generated by a run of either $M_0$ or $M_1$ randomly selected, or (ii) by interrogating an {\em oracle} (randomly) programmed with either $M_0$ or $M_1$ and then inspecting the related output.
We focus on the latter.
Informally, an oracle is an entity that provides {\em black-box} access to functionalities.
So in our distinguishing game, the adversary chooses inputs to the functionality exposed by the oracle, which can be called by the adversary and always provides the corresponding output.
Definitions involving oracles can be actually captured by providing the adversary with the {\em full} description of the functions, as an input, along with the transcripts of the protocol.
However, they have been introduced in cryptography when a such full description of the functions may have super-polynomial length (e.g. a random function~\cite{katz2014introduction}), and the (polynomial) adversary would not even have enough time to read its input.

Importantly, computational indistinguishability allows the adversary to win the distinguishing experiment, illustrated in Figure~\ref{fig:oracle-exp-dist}, by at most a negligible probability.
However, in this paper we show a case of perfect indistinguishability, i.e. the special case where the negligible probability is zero.
As such, the adversary is described by a PPT algorithm, and definitions must hold for any such adversary.
The same complexity restriction is reflected to oracle definitions, where the adversary is limited to interrogating the oracle at most a polynomial number of times.

\subsection{Eager-lazy approach}
\label{sec:eager-lazy-approach}
The eager-lazy approach was formally introduced by Bellare and Rogaway~\cite{bellare2004game} as a sampling technique to reason about lazy constructions by relating them to their corresponding eager constructions where the random variables are all sampled upfront.
They used their technique to prove indistinguishability in cryptographic games~\cite{bellare2006security}.
To appreciate the non-triviality of such an approach, we discuss the core lemma that allows one to substitute an already sampled value with a resampled one.
The following lemma is by Barthe et al.~\cite{barthe2009formal} with the notation we used in Section~\ref{sec:lazy-sampling-reasoning-pRHL}, as introducing their exact notation would unnecessarily require additional sections\footnote{Our notation is almost the same, we needed to reserve some letters as we are merging with lazy-sampling labeling notation.}.
\begin{lemma}[Lazy/eager sampling]
\label{th:eager-lazy-approach}
  Let $C[\bullet]$ be a context, $\ecmem_1$ and $\ecmem_2$ the memories in the leftmost and rightmost programs in a \prhl judgment, $c_1$ and $c_2$ commands, $e$ a boolean expression, $\delta$ a distribution expression, and $v$ a variable such that $C[\bullet]$ does not modify $FV\of{e} \cup FV\of{\delta}$\footnote{$FV$ are the free variables in an expression.} and does not use $v$.
  Assume
  \[
    \begin{split}
    |= &  \begin{array}{rcl}
            v \getr \delta; c_1; \mathbf{if}\ e\ \mathbf{then}\ v \getr \delta & \sim & v \getr \delta; c_1 \\
          \end{array} \\
    :\ & \Psi \land e => e \land =_v, \\
    \end{split}
  \]
  where $\Psi$ states that the values used in $c_1$ and $c_2$ are equivalent, and
  \[
    |=  c_2 \sim c_2:\ \Psi \land \lnot e => \lnot e \land =_v. \\
  \]
  Let $c$ and $c'$ be
  \[
    \begin{split}
       c & = \mathbf{if}\ e\ \mathbf{then}\ \cbr{v \getr \delta; c_1}\ \mathbf{else}\ c_2 \\
      c' & = \mathbf{if}\ e\ \mathbf{then}\ c_1\ \mathbf{else}\ c_2
    \end{split}
  \]
  then
  \[
    \begin{split}
    |= &  \begin{array}{rcl}
            C[c]; \mathbf{if}\ e\ \mathbf{then}\ v \getr \delta & \sim & v \getr \delta; C[c'] \\
          \end{array} \\
    :\ & \Psi \land e =>\ =_v. \\
    \end{split}
  \]
\end{lemma}
Where contexts can be seen as efficient algorithms that deduce an expression from another.
If we refer to the introduction, we see that a second artificial sample operation is inserted in one of the constructions: in other words, this allows to prove that $P_1 \compeq \eager_1$.
We remark that the above programs show only the interesting statements which appear in the crucial proof obligation under study.
They apply in general, for any two programs that include those lines to be compared, regardless of what are their next and previous statements in either program.
The intuitive discussion that they give is the following.
\begin{quotation}
\small
``
  In the above lemma $e$ indicates whether $v$ has not been used in the game since it was last sampled.
  If it has not been used, then it is perfectly fine to resample it.
  The first two hypotheses ensure that $e$ has exactly this meaning, $c_1$ must set it to {\tt false} if it has used the value sampled in $v$, and $c_2$ must not reset $e$ if it is {\tt false}.
  The first hypothesis is the one that allows swapping $c_1$ with $v \getr \delta$, provided the value of $v$ is not used in $c_1$.
  Note that, for clarity, we have omitted environments in the above lemma, and so the second hypothesis is not as trivial as it may seem because both programs may have different environments.
''~\cite{barthe2009formal}
\end{quotation}
We stress that the context $C[\bullet]$ and the environments in the above lemma are generally complex in the task of proving the indistinguishability between a lazy construction and its corresponding eager construction, and they purposely omit further details.
In a nutshell, the complexity of the eager-lazy approach is in the need to introduce dozens of extra eager-tactics to allow interprocedural code motion, overcoming the limitations of the tactic $\tactic{swap}$.
Proofs would become even more intricate in the case of directly proving indistinguishability between two lazy constructions like $P_1$ and $P_2$.
This extra complexity is usually tamed by creating additional intermediate games~\cite{shoup2004sequences}.
We solve the same problem by allowing the cryptographer complete proofs with oracles without the need for interprocedural code motion.

\subsection{Some EasyCrypt code}
\label{sec:ec-code}

This section shows some \EasyCrypt code that in the paper is shown in pseudocode, or mathematical formulas.

\subsubsection{Labeled types}
the labeled types, predicates, and operations illustrated in Section~\ref{sec:ifsupport} are implemented with the code listed in Fig.~\ref{fig:ec-if}.
\begin{figure}[htbp]
  \caption{Implementation of the labeled types, predicates, and operations}
  \label{fig:ec-if}
  \begin{lstlisting}[language=easycrypt,basicstyle=\normalsize\ttfamily,print]
type confidentiality = [ SECRET | LEAKED ].
type 'a leakable
  = 'a * ('a distr) option * confidentiality.
op is_secret ['a] (v: 'a leakable)
  = SECRET = v.`3.
op is_leaked ['a] (v: 'a leakable)
  = !(is_secret v).
op sampled_from ['a]
  (d: 'a distr) (v: 'a leakable)
  = v.`2 = Some d.
op ovd_eq ['a] (v w: ('a leakable) option)
  = ((oget v).`1, (oget v).`2)
  = ((oget w).`1, (oget w).`2).
abbrev (===) ['a]
  (v w: ('a leakable) option) = ovd_eq v w.
abbrev (<=)
  ['a] v d = sampled_from d v.
  \end{lstlisting}
\end{figure}

\subsubsection{The two lazy constructions}
\label{sec:lazy-constructions-ec}
the pseudo-code in Fig.~\ref{fig:if-example} reflects the actual \EasyCrypt code in Fig.~\ref{fig:ec-if-example}.
\begin{figure}[htbp]
  \caption{Implementation in \EasyCrypt of the two lazy constructions $P_1$ and $P_2$ using lazy-sampling labeling, compare with Figure~\ref{fig:if-example}.}
  \label{fig:ec-if-example}
    \mathligsoff
    \begin{lstlisting}[language=easycrypt,basicstyle=\normalsize\ttfamily,print]
module P1: RF = {
  var t: (X, Y leakable) fmap
  proc init() = { t = empty; }
  proc f(x: X): Y = {
    var ret: Y;
    if (!(dom t x)) {
      t.[x] </$ dv;
    }
    ret </ oget t.[x]; return ret;
  }

  proc g(x: X) = { (* do nothing *) }
}.
    \end{lstlisting}
    \begin{lstlisting}[language=easycrypt,basicstyle=\normalsize\ttfamily,print]
module P2: RF = {
  var t: (X, Y leakable) fmap
  proc init() = { t = empty; }
  proc f(x: X): Y = {
    var ret: Y;
    if (!(dom t x)) {
      t.[x] </$ dv;
    }
    ret </ oget t.[x]; return ret;
  }

  proc g(x: X) = {
    if (!(dom t x)) {
      t.[x] </$ dv;
    }
  }
}.
    \end{lstlisting}
\end{figure}

\subsection{Indistinguishability through the experiment.}
\label{sec:theorem-experiment}
Even though the Theorem~\ref{th:sample-indistinguishability} correctly captures the indistinguishability of two constructions $P_1$ and $P_2$, its link with the cryptographic experiment illustrated in Figure~\ref{fig:oracle-exp-dist} might be non-intuitive.
We show the steps to reach the conclusion in the theorem by directly referring to the experiment $\experiment{Exp}^{\dist, \oracle}_{P_1,P_2}$ and its pseudo-code shown in Figure~\ref{fig:oracle-exp-dist-pseudocode}.

Firstly, we start showing an implementation of the experiment  in \EasyCrypt is as follows:
\begin{lstlisting}[language=easycrypt,basicstyle=\normalsize\ttfamily,print]
module Exp(P1: RF, P2: RF, D: Dist) = {
  proc main() = {
    var b, b';

    P1.init();
    P2.init();
    b <$ {0,1};
    if (b) {
      b' <@ D(P1).distinguish();
    } else {
      b' <@ D(P2).distinguish();
    }

    return b = b';
  }
}.
\end{lstlisting}
The experiment is won if it returns $1$, that is the distinguisher guessed $b$ correctly as $b' = b$.
Following the experiment, if the two constructions $P_1$ and $P_2$ are indistinguishable, then the distinguisher cannot win the experiment with a probability {\em significantly} higher than a coin toss.
Thus in general, we would expect a security theorem whose conclusion requires to prove that:
\[
  \prob{\experiment{Exp}^{\dist, \oracle}_{P_1,P_2} = 1} = \dfrac{1}{2} \pm \negligible[n].
\]

In the particular case of $P_1$ and $P_2$ as described in Figure~\ref{fig:if-example}, we can prove perfect indistinguishability, that is:
\[
  \prob{\experiment{Exp}^{\dist, \oracle}_{P_0,P_1} = 1} = \dfrac{1}{2}.
\]

It is easy to see that {\em unfolding} the definition of the experiment itself, we would reduce to prove that
\begin{itemize}
  \item $1 = \dist\of{\oracle_1}.\mathsf{dist}()$, in the case $b = 1$, and
  \item $0 = \dist\of{\oracle_2}.\mathsf{dist}()$, $b = 0$.
\end{itemize}
They are independent cases with probability of success equivalent to $\frac{1}{2}$, that multiplied with the probability of $b = 0$ or $b = 1$, that is $\frac{1}{2}$, they will eventually sum up to $\frac{1}{2}$ as expected.
In this particular case, it is immediate to see that the probability of guessing correctly is equivalent to the probability of guessing incorrectly, as they are complementary events:
\[
  \prob{\dist\of{\oracle_2} = 0} = \prob{\dist\of{\oracle_2} = 1} = \frac{1}{2}.
\]
In other words, we can reduce the perfect indistinguishability to the equivalence of the two probabilities (i) of guessing correctly in the case that the oracle $\oracle$ is programmed with $P_1$, and (ii) of guessing incorrectly in the case that the oracle $\oracle$ is programmed with $P_2$.
Formally,
\[
  \prob{\dist\of{\oracle_1} = 1} = \prob{\dist\of{\oracle_2} = 0},
\]
which is exactly the conclusion of the Theorem~\ref{th:sample-indistinguishability}.
We also note that the initialization procedure $init$, that sets the internal map as an empty map, is called by the experiment before the distinguisher.
This justifies the assumption of having empty maps in the Theorem~\ref{th:sample-indistinguishability}.

\subsection{Reproducing the problem from a sound situation}
\label{sec:reproducing-virtual-swap-rnd}
The problem we solve in the paper can be identified by the following judgment
\[
  |=
  \pwhileprog{v \getr \delta; r <- v}
  \sim
  \pwhileprog{r <- v}
  :\ =_x => =_r
\]
that is generally not true.
We show how to reproduce it by starting from a similar situation that is, conversely, provable using already existing and well supported tactics in \EasyCrypt{}: $\tactic{rnd}$ and $\tactic{swap}$.
In particular,
\[
  |= \pwhileprog{t\of{x} \rnd{\secasgn} \delta; r \secasgn t\of{x}} \sim \pwhileprog{c; t\of{x} \rnd{\secasgn} \delta; c'; r \secasgn t\of{x}}: =_x => =_r.
\]
where $c$ and $c'$ do not modify the value of $x$ nor access (read or write) $t\of{x}$ even indirectly: for the purpose of the post-condition, $c$ and $c'$ represent {\em dead code}.
This can be transformed using the tactic $\tactic{swap}$ to move the sampling just above the assignment, and discard the dead code\footnote{This uses another tactic, $\tactic{kill}$, but we did not want to overload the narrative.}.
We get two identical programs
\[
  |= \pwhileprog{t\of{x} \rnd{\secasgn} \delta; r \secasgn t\of{x}} \sim \pwhileprog{t\of{x} \rnd{\secasgn} \delta; r \secasgn t\of{x}}:\ =_x => =_r.
\]
At this point, it is clear that upon the same inputs ($=_x$) we will get the same outputs ($=_r$) and, therefore, the same distribution of their memories.
As we focus on reproducing our situation of interest, we can continue by de-sugaring the new syntax required to handle labeled values by using the tactic $\tactic{declassify}$, then the samplings by $\tactic{secrnd}$.
\[
  |=
  \pwhileprog{v \getr \delta; t\of{x} <- \triplet{v}{\delta}{\secret}; t\of{x} <- \triplet{v}{\delta}{\leaked}; r <- v}
  \sim
  \pwhileprog{v \getr \delta; t\of{x} <- \triplet{v}{\delta}{\secret}; t\of{x} <- \triplet{v}{\delta}{\leaked}; r <- v}
  :\ =_x => =_r
\]
The updates to the memory $m$ are irrelevant for the post-condition, as they only affect the {\em confidentiality} dimension; so, removing the deadcode, we get
\[
  |=
  \pwhileprog{v \getr \delta; r <- v}
  \sim
  \pwhileprog{v \getr \delta; r <- v}
  :\ =_x => =_r
\]
At this point, we can consume the first sampling with the tactic $\tactic{rnd}$ sided to the rightmost program only\footnote{Again, we actually need to apply another tactic first, $\tactic{seq}$, as $\tactic{rnd}$ is limited to consume the bottom statements in programs.}, and finally we get
\begin{equation}
  \label{eq:lemma-conclusion-consumed}
  |=
  \pwhileprog{v \getr \delta; r <- v}
  \sim
  \pwhileprog{r <- v}
  :\ =_x => =_r
\end{equation}
According to its semantics, see Fig.~\ref{fig:pwhile-semantics}, we know that the rightmost memory has been updated with
\begin{equation}
  \label{eq:lemma-premise}
  \bind \of{\pwhileeval{\delta} \ecmem_2} \of{\lambda f. \return \ecmem_2\rwrule{f}{v}},
\end{equation}
even if such a relation formally disappeared from Eq.~\ref{eq:lemma-conclusion-consumed}.
This perfectly reproduces the situation we are interested in.

We also notice that up to this point, we could fully mechanize the whole reasoning in \EasyCrypt{}.
Another way to see it is that the choice of tactics in the last step, even if possible, does not represent a viable strategy to finish the proof.
This is because tactics can (soundly) weaken preconditions with the risk of not leaving enough hypotheses to complete the proof -- and incomplete proofs do not jeopardize the soundness of theorem provers, e.g., $P |- Q$ clearly entails $P => Q, P |- Q$ ({\em modus ponens}), but the former is not generally true.
We notice that we could very well complete our nested proof in \EasyCrypt if we had used the $\tactic{rnd}$ to contemporarily consume the sampling in both programs; however, this would not have helped our proof.

\subsection{Formal proofs}
\label{sec:formal-proofs}

\begin{lemma}[Virtual $\tactic{swap/rnd}$]
  Let $\ecmem_1$ and $\ecmem_2$ be the memories in the leftmost and rightmost programs in a \prhl judgment, $\delta$ a (lossless) distribution expression, and $t$ a map and equally defined in both memories.
  Assume that
  \[
    \pwhileeval{t\of{x}} \ecmem_2 = \bind \of{\pwhileeval{\delta} \ecmem_2} \of{\lambda f. \return \ecmem_2\rwrule{f}{t\of{x}}}.
  \]
  then we have
  \[
    |= t\of{x} \getr \delta; r <- t\of{x}  \sim  r <- t\of{x}: =_x => =_r.
  \]
\end{lemma}
\begin{proof}
The conclusion of our statement is not generally true.
It is false, for example, in the case where, in the rightmost program, the value of $t\of{x}$ is deterministic.
However, our hypothesis excludes such cases and we can complete the proof outside of \EasyCrypt, as we need to directly reason about the internal structure of the memory.
To proceed with our proof, we can still use \EasyCrypt's tactics, but need to annotate the related memory updates aside on a piece of paper.

We can start transforming the conclusion of the lemma by applying the tactic $\tactic{assign}$, see Fig.~\ref{fig:ec-tactics}.
Its effect is to replace the equality of $r$ in the post-condition directly with $t\of{x}$, as the latter will be (deterministically) associated to $r$ in the memory: we obtain
\[
  |=
  \pwhileprog{t\of{x} \getr \delta}
  \sim
  \pwhileprog{}
  :\ =_x => =_{t\of{x}}.
\]
We then can apply (only) the tactic $\tactic{rnd}$ sided to the leftmost program (see Fig.~\ref{fig:ec-tactics}) to discharge probabilistic assignments.
Speaking about the distribution over memory, it translates to require that the distribution $\delta$ is lossless (it is by hypothesis).%
At the same time, reflecting the semantics of the probabilistic assignment in Fig.~\ref{fig:pwhile-semantics}, we annotate
\[
  \bind \of{\pwhileeval{\delta} \ecmem_1} \of{\lambda f. \return \ecmem_1\rwrule{f}{t\of{x}}}.
\]
As prefigured, we have no way to complete such a proof in \EasyCrypt, as the relations over the distribution over memories (related to $t\of{x}$) disappeared.
Yet, with the aid of our annotation on paper, we know that values of $t\of{x}$ in the memories $\ecmem_1$ and $\ecmem_2$ match, as per our hypothesis, and the proof is complete.
\end{proof}

\begin{lemma}[Virtual $\tactic{secrndasgn}$]
  Let $\ecmem_1$ and $\ecmem_2$ be the memories in the leftmost and rightmost programs in a \prhl judgment, $\delta$ a distribution expression, and $t$ a map {\em with labeled codomain} and equally defined in both memories.
  Assume that
  \[
    \pwhileeval{t\of{x}} \ecmem_2 = \bind \of{\pwhileeval{\delta} \ecmem_2} \of{\lambda f. \return \ecmem_2\rwrule{\triplet{f}{\delta}{H}}{t\of{x}}}.
  \]
  then we have
  \[
    |= t\of{x} \rnd{\secasgn} \delta; r \secasgn t\of{x}  \sim  r \secasgn t\of{x}: =_x => =_r.
  \]
\end{lemma}
\begin{proof}
We start our proof de-sugaring the new syntax by using the tactics $\tactic{declassify}$ and $\tactic{secrnd}$, we obtain the leftmost program
\[
  {v \getr \delta; t\of{x} <- \triplet{v}{\delta}{\secret}; t\of{x} <- \triplet{v}{\delta}{\leaked}; r <- v}
\]
and the rightmost program
\[
  {t\of{x} <- \triplet{\proj{1}{t\of{x}}}{\proj{2}{t\of{x}}}{\leaked}; r <- \proj{1}{t\of{x}}}.
\]
If we rewrite the rightmost posing $v\inecmem{\ecmem_2} \eqdef \proj{1}{t\of{x}}\inecmem{\ecmem_2}$ (it does not affect $v\inecmem{\ecmem_1}$), we get
\[
  |=
  \pwhileprog{v \getr \delta; t\of{x} <- \triplet{v}{\delta}{\secret}; t\of{x} <- \triplet{v}{\delta}{\leaked}; r <- v}
  \sim
  \pwhileprog{t\of{x} <- \triplet{v}{\proj{2}{t\of{x}}}{\leaked}; r <- v}
\]
and note that, again, the updates to memories are dead code.
Removing the dead code, we obtain
\[
  |=
  \pwhileprog{v \getr \delta; r <- v}
  \sim
  \pwhileprog{r <- v}
  :\ =_x => =_r.
\]
The above \prhl judgment matches with Eq.~\ref{eq:lemma-conclusion-consumed}.

The hypothesis of this lemma shows how a triplet is stored in the memory location pointed by $t\of{x}$.
If we isolate the values apart, we can write
\[
  \begin{array}{l}
    \pwhileeval{\proj{1}{t\of{x}}} \ecmem_2 = \bind \of{\pwhileeval{\delta} \ecmem_2} \of{\lambda f. \return \ecmem_2\rwrule{{f}}{\proj{1}{t\of{x}}}}, \\
    \pwhileeval{\proj{2}{t\of{x}}} \ecmem_2 = \return \ecmem_2\rwrule{\pwhileeval{\delta} \ecmem_2}{\proj{2}{t\of{x}}}, \\
    \pwhileeval{\proj{3}{t\of{x}}} \ecmem_2 = \return \ecmem_2\rwrule{\pwhileeval{\secret} \ecmem_2}{\proj{3}{t\of{x}}}. \\
  \end{array}
\]
The last two are not relevant (and independent) to the conclusion; conversely, the first matches with the hypothesis in Eq.~\ref{eq:lemma-premise}.
Thus, we reduced our proof to that of Lemma~\ref{th:virtual-swap-rnd} and our proof is complete.
\end{proof}

\begin{lemma}[De-virtualized $\tactic{secrndasgn}$ hypothesis]
  Let $\ecmem_1$ and $\ecmem_2$ be the memories in the leftmost and rightmost programs in a \prhl judgment, $\delta$ a (lossless) distribution expression, $t$ a map with {\em labeled codomain} and equally defined in both memories, and $\invariant$ the invariant as specified in Eq.~\ref{eq:invariant-maps}.
  If we call
  \[
    \begin{split}
    h_1 & \eqdef =_x, h_2 \eqdef \invariant\of{t\inecmem{\ecmem_1}, t\inecmem{\ecmem_2}, \delta}, \\
    h_3 & \eqdef x \notin t_X\inecmem{\ecmem_1}, \text{ and } h_4 \eqdef x \in t_X\inecmem{\ecmem_2}, \\
    \end{split}
  \]
  then
  \begin{equation}
    \label{eq:practical-secrndasgn-hyps}
    h_1 \land h_2 \land h_3 \land h_4
  \end{equation}
  implies that
  \[
    \pwhileeval{t\of{x}} \ecmem_2 = \bind \of{\pwhileeval{\delta} \ecmem_2} \of{\lambda f. \return \ecmem_2\rwrule{\triplet{f}{\delta}{H}}{t\of{x}}}.
  \]
\end{lemma}
\begin{proof}
In brief, we need to show that conditions on labeled memories uniquely and unequivocally capture the conditions over the distribution over memories (binding generated by a sampling).
For simplicity, we denote $m\inecmem{\ecmem_1}$ as $t$ and $m\inecmem{\ecmem_2}$ as $t'$.
We start our proof specializing the invariant $h_2$ with $x$ (and its equality in both memories as stated in $h_1$); we obtain
\[
  \begin{split}
  h_2 =\ & \left( x \in t'_X => t'\of{x} \inr \delta                                   \right) \\
    & \land\; \left( x \in t_X => x \in t'_X \land \proj{1}{t\of{x}} = \proj{1}{t'\of{x}} \right) \\
    & \land\; \left( x \in t_X => \isleaked {t\of{x}} => t\of{x} = t'\of{x}              \right) \\
    & \land\; \left( x \notin t_X => x \in t'_X => \lnot\isleaked {t'\of{x}}              \right). \\
  \end{split}
\]
We proceed simplifying the invariant $h_2$ with $h_3$ and $h_4$:
\[
  \begin{split}
  h_2 =\ & \left( x \in t'_X => t'\of{x} \inr \delta                                   \right) \\
    & \land\; \left( x \notin t_X => x \in t'_X => \lnot\isleaked {t'\of{x}}              \right) = \\
    =\ & t'\of{x} \inr \delta \land \lnot\isleaked {t'\of{x}}, \\
  \end{split}
\]
that for convenience we split it into two and rewrite its notation according to the definition in Section~\ref{sec:ifsupport}; we have:
\[
  h_2' \eqdef \proj{2}{t'\of{x}} = \delta, \quad h_2'' \eqdef \proj{3}{t'\of{x}} = \secret.
\]
First we analyze $h_2''$, which assumes that the confidentiality label of $t'\of{x}$ is $\secret$ (secret).
Syntactically, it is forbidden to assign to labeled values, and \EasyCrypt does not assign one automatically upon declaration.
Semantically, the {\em only} way to programmatically have a secret variable is to discharge the syntax of $\secrnd$ using the tactic $\tactic{secrnd}$.
This means that at some point, even in unpredictable calls to oracle functions, the statement $t'\of{x} \secrnd \delta'$ has been called.
Visibly, $h_2''$ does not suffice to prove our memory update, as the distribution can still be different from the one in the conclusion.

Now we use $h_2'$, which assumes that the distribution associated with $t'\of{x}$ is $\delta$.
Again syntactically, it is impossible to generate this hypothesis; and semantically there must have been either a statement using the syntax $\secasgn$, or a statement using the syntax $\secrnd$.
In the former case, we would have $\proj{3}{t'\of{x}} = \leaked$, which is in contradiction with $h_2''$.
So the only viable case is that the statement must have been $\secrnd$.
Combining $h_2'$ and $h_2''$, we can {\em undoubtedly} be sure that $\delta' = \delta$, so the following statement has been consumed by $\tactic{secrnd}$: $t'\of{x} \secrnd \delta$.
Using semantics in Fig.~\ref{fig:pifwhile-semantics}, we can conclude that
\[
  \pwhileeval{t\of{x}} \ecmem_2 = \bind \of{\pwhileeval{\delta} \ecmem_2} \of{\lambda f. \return \ecmem_2\rwrule{\triplet{f}{\delta}{H}}{t\of{x}}}.
\]
We deduced our conclusion and the proof is complete.
\end{proof}

\begin{theorem}[Lazy/IF sampling]
  Let $\ecmem_1$ and $\ecmem_2$ be the memories in the leftmost and rightmost programs in a \prhl judgment, $\delta$ a (lossless) distribution expression, $t$ a map with {\em labeled codomain} and equally defined in both memories, and $\invariant$ the invariant as specified in Eq.~\ref{eq:invariant-maps}.
  Assume that
  \begin{equation}
    \label{eq:if-approach-hyps}
    \begin{split}
    \precondition =\ & =_x \land\ \invariant\of{t\inecmem{\ecmem_1}, t\inecmem{\ecmem_2}, \delta} \\
    & \land x \notin m_X\inecmem{\ecmem_1} \land\ x \in m_X\inecmem{\ecmem_2}, \\
    \end{split}
  \end{equation}
  and let the post-condition $\postcondition$ be
  \[
    \postcondition = (=_r)\ \land t\of{x}\inecmem{\ecmem_1} \inr \delta\ \land \invariant\of{t\inecmem{\ecmem_1}, t\inecmem{\ecmem_2}, \delta},
  \]
  then we have
  \[
    |= t\of{x} \rnd{\secasgn} \delta; r \secasgn t\of{x}  \sim  r \secasgn t\of{x}: \precondition => \postcondition.
  \]
\end{theorem}
\begin{proof}
We can split the proof into three parts
\[
  \begin{array}{l}
    \postcondition_1 \eqdef \pbr{=_r} \\
    \postcondition_2 \eqdef t\of{x}\inecmem{\ecmem_1} \inr \delta = \proj{2}{t\of{x}}\inecmem{\ecmem_1} = \delta \\
    \postcondition_3 \eqdef \invariant\of{t\inecmem{\ecmem_1}, t\inecmem{\ecmem_2}, \delta} \\
  \end{array}
\]

{\bf First part: same output.}
Proving the first part, $\postcondition_1$,  is straightforward, we notice that our precondition $\precondition$ matches that of Lemma~\ref{th:practical-secrndasgn}, that can provide us with the hypothesis of Lemma~\ref{th:virtual-secrndasgn}.
We also notice that our \prhl judgment's precondition is stronger than the one in Lemma~\ref{th:virtual-secrndasgn}, so the lemma can be applied and proves our (first) conclusion.

{\bf Second part: intended distribution.}
To prove the second part, $\postcondition_2$, we de-sugar the new syntax using the tactics $\tactic{secasgn}$ and $\tactic{secrnd}$:
\[
  |=
  \pwhileprog{v \getr \delta; t\of{x} <- \triplet{v}{\delta}{\secret}; \proj{3}{t\of{x}} <- \leaked; r <- v}
  \sim
  \pwhileprog{\proj{3}{t\of{x}} <- \leaked; r <- v}
  : \precondition => \postcondition_2
  .
\]
We notice that, relating to $\postcondition_2$, the whole rightmost program is dead code, so are the first two statements in the leftmost program; removing the dead code we obtain:
\[
  |=
  \pwhileprog{t\of{x} <- \triplet{v}{\delta}{\leaked}}
  \sim
  \pwhileprog{}
  : \precondition => \postcondition_2
\]
The leftmost program trivially entails $\postcondition_2$.

{\bf Third part: invariant.}
To prove this last part, $\postcondition_3$, we start remarking that the invariant {\em after} the execution of the programs deals with an updated map, $m\inecmem{\ecmem_1'}$, at index $x$.
So, all the rest of the maps must still be behaving according to the invariant, that is
\[
  \forall x'. x' \neq x => \invariant\of{t\inecmem{\ecmem_1'}, t\inecmem{\ecmem_2}, \delta}.
\]
The only part of memory that could be problematic is exactly at the index $x$, for that we have to prove that:
\[
  \begin{split}
  i \eqdef & \left( x \in t'_X => t'\of{x} \inr \delta                                   \right) \\
    & \land\; \left( x \in t_X => x \in t'_X \land \proj{1}{t\of{x}} = \proj{1}{t'\of{x}} \right) \\
    & \land\; \left( x \in t_X => \isleaked {t\of{x}} => t\of{x} = t'\of{x}              \right) \\
    & \land\; \left( x \notin t_X => x \in t'_X => \lnot\isleaked {t'\of{x}}              \right). \\
  \end{split}
\]
where for simplicity, again we denoted $t\inecmem{\ecmem_1'}$ as $t$ and $t\inecmem{\ecmem_2}$ as $t'$.
After the execution we know that $x \in t_X$ and $x \in t'_X$; so $i$ becomes
\[
  \begin{split}
  i =\ & \left( x \in t'_X => t'\of{x} \inr \delta                                   \right) \\
    & \land\; \left( x \in t_X => x \in t'_X \land \proj{1}{t\of{x}} = \proj{1}{t'\of{x}} \right) \\
    & \land\; \left( x \in t_X => \isleaked {t\of{x}} => t\of{x} = t'\of{x}              \right) = \\
    =\ & t'\of{x} \inr \delta \land\; \proj{1}{t\of{x}} = \proj{1}{t'\of{x}} \\
    & \land\; \left( \isleaked {t\of{x}} => t\of{x} = t'\of{x}              \right) \\
    =\ & \proj{2}{t'\of{x}} = \delta \land\; \proj{1}{t\of{x}} = \proj{1}{t'\of{x}} \\
    & \land\; \left( \proj{3}{t\of{x}} = \leaked => t\of{x} = t'\of{x}              \right). \\
  \end{split}
\]
So the post-condition $\postcondition_3$ can be split into two parts
\[
  \postcondition_3' = \forall x'. x' \neq x => \invariant\of{t\inecmem{\ecmem_1'}, t\inecmem{\ecmem_2}, \delta}, \quad
  \postcondition_3'' = c.
\]
Proving $\postcondition_3'$ is trivial and we are left to prove $\postcondition_3''$; we de-sugar the new syntax of the programs using the tactics $\tactic{secasgn}$ and $\tactic{secrnd}$:
\[
  |=
  \pwhileprog{v \getr \delta; t\of{x} <- \triplet{v}{\delta}{\secret}; \proj{3}{t\of{x}} <- \leaked; r <- v}
  \sim
  \pwhileprog{\proj{3}{t\of{x}} <- \leaked; r <- v}
  : \precondition => \postcondition_3''
  .
\]
With respect to the post-condition $\postcondition_3''$ the last statement of both programs is irrelevant; removing dead code, we get
\[
  |=
  \pwhileprog{v \getr \delta; t\of{x} <- \triplet{v}{\delta}{\secret}; \proj{3}{t\of{x}} <- \leaked}
  \sim
  \pwhileprog{\proj{3}{t\of{x}} <- \leaked}
  : \precondition => \postcondition_3''
  .
\]
After the execution of the above statements, it is clear that $t\of{x}$ is labeled as leaked, $\leaked$ in both memories, but this is only the third element of the tuple defining the value.
The above execution does not make it possible to prove (or disprove) whether their value, the first element of the tuple, or their distribution, the second element of the tuple, are the same.
However, we notice that we already proved $\postcondition_1$ and $\postcondition_2$, so we know that ${t\of{x}}\inecmem{\ecmem_1} \simeq {t\of{x}}\inecmem{\ecmem_2}$ since
\[
  \begin{array}{l}
    \proj{1}{t\of{x}}\inecmem{\ecmem_1} = \proj{1}{t\of{x}}\inecmem{\ecmem_2} = \proj{1}{M\of{x}} \\
    \proj{2}{t\of{x}}\inecmem{\ecmem_1} = \proj{2}{t\of{x}}\inecmem{\ecmem_2} = \delta \\
  \end{array}
\]
Thus, along with the confidentiality label, we can state $t\of{x} = t'\of{x}$.
The above result trivially entails $\postcondition_3''$ and the proof is complete.
\end{proof}
 
\end{document}